\documentclass[12pt,reqno]{amsart}
\usepackage{hyperref}
\usepackage{amsmath,amssymb,amsthm}
\usepackage{mathtools}
\usepackage{stmaryrd}
\usepackage{color}
\usepackage{color}
\usepackage{hyperref}
\usepackage{url}
\usepackage{breakurl}
\usepackage{comment}
\usepackage{mathtools}
\newcommand{\bburl}[1]{\textcolor{blue}{\url{#1}}}
\usepackage[inline]{enumitem}
\usepackage{graphicx,setspace}
\usepackage{tikz}
\usetikzlibrary{arrows}
\usetikzlibrary{positioning,shadows,arrows.meta}
\usepackage{wasysym}
\usepackage[makeroom]{cancel}
\usepackage{fullpage}
\usepackage{comment}
\usepackage{enumitem}
\usepackage{caption}
\allowdisplaybreaks

\theoremstyle{plain}
\newtheorem{theorem}{Theorem}[section]
\newtheorem{proposition}[theorem]{Proposition}

\newtheorem{assumption}[theorem]{Assumption}
\newtheorem{lemma}[theorem]{Lemma}
\newtheorem{corollary}[theorem]{Corollary}

\newtheorem{remark}[theorem]{Remark}
\newtheorem*{claim*}{Claim}
\theoremstyle{definition}
\newtheorem{example}[theorem]{Example}
\newtheorem*{notation*}{Notation}
\newtheorem{definition}[theorem]{Definition}
\providecommand{\abs}[1]{\left\lvert#1\right\rvert}
\providecommand{\card}[1]{\left\lvert#1\right\rvert}
\providecommand{\norm}[1]{\left\lvert\left\lvert#1\right\rvert\right\rvert}
\providecommand{\bigO}[1]{O\left(#1\right)}
\providecommand{\littleO}[1]{o\left(#1\right)}

\def\E{\mathbb{E}}

\def\N{\mathbb{N}}

\def\R{\mathbb{R}}
\def\C{\mathbb{C}}

\def\P{\mathbb{P}}

\newcommand{\ncr}[2]{{#1 \choose #2}}

\newcommand\be{\begin{equation}}
\newcommand\ee{\end{equation}}
\newcommand\bea{\begin{eqnarray}}
\newcommand\eea{\end{eqnarray}}

\numberwithin{equation}{section}

\date{\today}

\begin{document}

%----------------------------------%

\title[Spectral Statistics of Non-Hermitian Random Matrix Ensembles]{Spectral Statistics of Non-Hermitian Random Matrix Ensembles}

%----------------------------------%

\author{Ryan C. Chen}
\address[RCC]{\scriptsize Department of Mathematics, Princeton University, Princeton, NJ 08544}
\email{\textcolor{blue}{\href{mailto:rcchen@princeton.edu}{rcchen@princeton.edu}}}

\author{Yujin H. Kim}
\address[YHK]{\scriptsize Department of Mathematics, Columbia University, New York, NY 10027}
\email{\textcolor{blue}
{\href{mailto:yujin.kim@columbia.edu}{yujin.kim@columbia.edu}}}

\author{Jared D. Lichtman}
\address[JDL]{\scriptsize Department of Mathematics, Dartmouth College, Hanover, NH 03755}
\email{\textcolor{blue}{\href{mailto:jdl.18@dartmouth.edu}{jdl.18@dartmouth.edu}}}

\author{Steven J. Miller}
\address[SJM]{\scriptsize Department of Mathematics and Statistics, Williams College, Williamstown, MA 01267}
\email{\textcolor{blue}{\href{mailto:sjm1@williams.edu}{sjm1@williams.edu}}}

\author{Shannon Sweitzer}
\address[SS]{\scriptsize Department of Mathematics, University California, Riverside, CA 92521}
\email{\textcolor{blue}{\href{mailto:sswei001@ucr.edu}{sswei001@ucr.edu}}}

\author{Eric Winsor}
\address[EW]{\scriptsize Department of Mathematics, University of Michigan, Ann Arbor, MI 48109}
\email{\textcolor{blue}{\href{mailto:rcwnsr@umich.edu}{rcwnsr@umich.edu}}}

\thanks{The authors were partially supported by NSF Grants DMS1561945, DMS1659037, the University of Michigan, Princeton University, and Williams College. We thank Arup Bose, Peter Forrester, Roger van Peski, and our colleagues from SMALL 2017 for helpful conversations. }



\begin{abstract} Recently Burkhardt et. al. introduced the $k$-checkerboard random matrix ensembles, which have a split limiting behavior of the eigenvalues (in the limit all but $k$ of the eigenvalues are on the order of $\sqrt{N}$ and converge to semi-circular behavior, with the remaining $k$ of size $N$ and converging to hollow Gaussian ensembles). We generalize their work to consider non-Hermitian ensembles with complex eigenvalues; instead of a blip new behavior is seen, ranging from multiple satellites to annular rings. These results are based on moment method techniques adapted to the complex plane as well as analysis of singular values. \end{abstract}

\maketitle

%---------------------------------%

\tableofcontents

%--------------------------------%
%%%%%%%%%%%%%%%%%%%%%%%%%%%%%%%%%%%%%%%%%%%%%%%%%%%%%%%%%%%%%%%%%%%%%%%%%%%%%%%%%%%%%%%%%%%%%%%%%%%%%%%%%%%%%%%%
%%%%%%%%%%%%%%%%%%%%%%%%%%%%%%%%%%%%%%%%%%%%%%%%%%%%%%%%%%%%%%%%%%%%%%%%%%%%%%%%%%%%%%%%%%%%%%%%%%%%%%%%%%%%%%%%
%%%%%%%%%%%%%%%%%%%%%%%%%%%%%%%%%%%%%%%%%%%%%%%%%%%%%%%%%%%%%%%%%%%%%%%%%%%%%%%%%%%%%%%%%%%%%%%%%%%%%%%%%%%%%%%%
%%%%%%%%%%%%%%%%%%%%%%%%%%%%%%%%%%%%%%%%%%%%%%%%%%%%%%%%%%%%%%%%%%%%%%%%%%%%%%%%%%%%%%%%%%%%%%%%%%%%%%%%%%%%%%%%
\section{Introduction}

%%%%%%%%%%%%%%%%%%%%%%%%%%%%%%%%%%%%%%%%%%%%%%%%%%%%%%%%%%%%%%%%%%%%%%%%%%%%%%%%%%%%%%%%%%%%%%%%%%%%%%%%%%%%%%
%%%%%%%%%%%%%%%%%%%%%%%%%%%%%%%%%%%%%%%%%%%%%%%%%%%%%%%%%%%%%%%%%%%%%%%%%%%%%%%%%%%%%%%%%%%%%%%%%%%%%%%%%%%%%%
%%%%%%%%%%%%%%%%%%%%%%%%%%%%%%%%%%%%%%%%%%%%%%%%%%%%%%%%%%%%%%%%%%%%%%%%%%%%%%%%%%%%%%%%%%%%%%%%%%%%%%%%%%%%%%
\subsection{Background}

Random matrix ensembles have been studied for almost a hundred years. The eigenvalues of these ensembles model many important and interesting behavior, from the waiting time of events to the energy levels of heavy nuclei to zeros of $L$-functions in number theory; see for example the surveys \cite{Bai, BFMT-B, Con, FM, KaSa, KeSn} and the textbooks \cite{Fo, Me, MT-B, Tao2011}.

There are many questions one can ask about these eigenvalues. This paper is a sequel to \cite{RMT2016}. There, in the spirit of numerous previous works, the authors investigated the density of eigenvalues of some highly structured ensembles. One of the central results in the subject is due to Wigner \cite{Wig1, Wig2, Wig3, Wig4, Wig5}, which states that the distribution of the scaled eigenvalues of a typical real symmetric matrix converges, in some sense, to the semi-circle distribution. However, if the real symmetric matrices have additional structure then other distributions can arise; see for example \cite{Bai, BasBo1, BasBo2, BanBo, BLMST, BCG, BHS1, BHS2, BM, BDJ, GKMN, HM, JMRR, JMP, Kar, KKMSX, LW, MMS, MNS, MSTW, McK, Me, Sch}.

In all those examples the limiting distribution has just one component. Different behavior is seen in the limit as $N\to\infty$ of the $k$-checkerboard $N\times N$ matrix ensembles of \cite{RMT2016} (see also \cite{CDF, CDF2}), described later in Definition \ref{def:complexSymCheck}. There, all but $k$ of the normalized eigenvalues converge to a semi-circle centered at the origin; however, there are $k$ eigenvalues which diverge to infinity together. Further, these $k$ blip eigenvalues converge to a universal distribution, the $k$-hollow GOE distribution (obtained by setting the diagonal of the $k\times k$ GOE ensemble to 0).

Below we describe the ensembles studied in \cite{RMT2016} and discuss our generalization (see Definitions \ref{def:gencheckone} and \ref{def:genchecktwo}). In particular, we find ensembles where there can be multiple blips or satellites orbiting the bulk of the eigenvalues, as well as a ring of eigenvalues around the central mass; Figure \ref{fig:imagestoseelater}.

\begin{figure}[h]
\scalebox{1}{\includegraphics{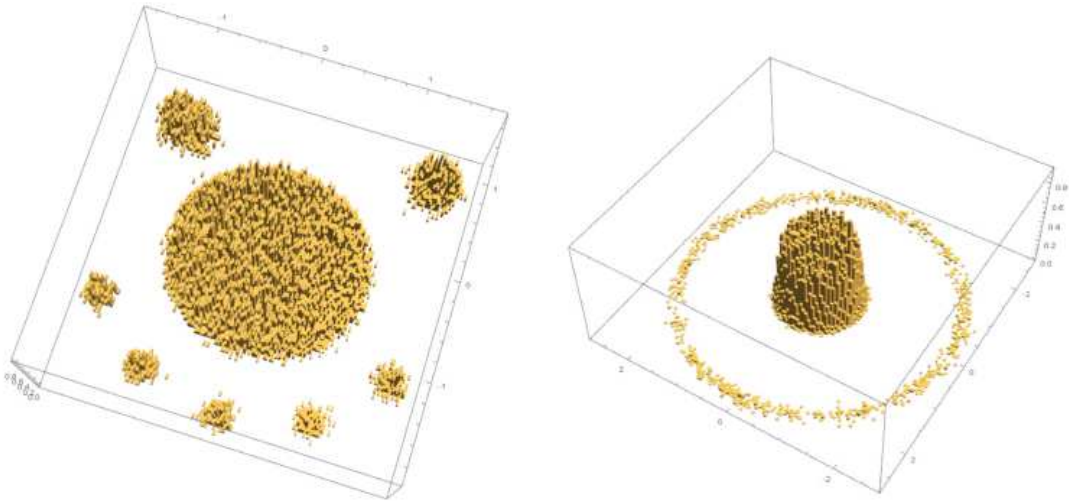}}
%\scalebox{.4}{\includegraphics{Smiley.eps}}\ \ \scalebox{.4}{\includegraphics{SaturnV.eps}}
\caption{Two numerical examples of distributions which can arise from a generalized $k$-checkerboard ensemble. Left: A collection of satellites. Right: A ring of eigenvalues.}
\label{fig:imagestoseelater}
\end{figure}

In the next subsections we define the ensembles we investigate and state our results. Unfortunately many of the techniques used for related ensembles are not applicable here, and thus we spend some time describing the needed tools and approach.

%%%%%%%%%%%%%%%%%%%%%%%%%%%%%%%%%%%%%%%%%%%%%%%%%%%%%%%%%%%%%%%%%%%%%%%%%%%%%%%%%%%%%%%%%%%%%%%%%%%%%%%%%%%%%%
%%%%%%%%%%%%%%%%%%%%%%%%%%%%%%%%%%%%%%%%%%%%%%%%%%%%%%%%%%%%%%%%%%%%%%%%%%%%%%%%%%%%%%%%%%%%%%%%%%%%%%%%%%%%%%
%%%%%%%%%%%%%%%%%%%%%%%%%%%%%%%%%%%%%%%%%%%%%%%%%%%%%%%%%%%%%%%%%%%%%%%%%%%%%%%%%%%%%%%%%%%%%%%%%%%%%%%%%%%%%%
\subsection{Results}

Random matrix ensembles with real entries see markedly different behavior between asymmetric and symmetric entry choices -- for example, the symmetric ensembles are Hermitian with real eigenvalues, and this need not hold for asymmetric ensembles. Allowing matrices with complex entries, we also find differences between the asymmetric and symmetric (not-necessarily Hermitian) ensembles in the joint density formulas.

Recall that the joint density function for singular values returns the probability that any given matrix has a certain $N$-tuple as its singular values.

Suppose $M$ is a random $N\times N$ matrix (for example, real asymmetric, complex symmetric, etc.). The joint density function $\rho_N$ for the singular values satisfies
\begin{align}
\int_{\R_{\geq0}^N} F(x_1,\ldots,x_N)\rho_N(x_1,\ldots,x_N)~dx\ = \ \E\sum_{\{\sigma_1^2,\ldots,\sigma_N^2\}\in\lambda(M^*M)} F(\sigma_1,\ldots,\sigma_N)
\end{align}
for any test function $F$, where the right-hand sum is interpreted as over all $N!$ orderings of the $N$ eigenvalues of $M^*M$ (and the $\sigma_j$ are nonnegative).

We list the available singular value joint density functions for complex asymmetric and symmetric ensembles, see for example \cite{AZ,TaoVu2009,Fo}:
\begin{align}
&{\rm Complex\ asymmetric\ Gaussian:}\ \rho_N(x_1,\ldots,x_N)\ =\ c_{N}'\left|\Delta(x_1^2,\ldots,x_N^2)\right|^2\prod_{j=1}^N |x_j| \prod_{j=1}^N e^{-|x_j|^2/2}
\end{align}

\ \\
\begin{align}
&{\rm Complex\ symmetric\ Gaussian:}\ \rho_N(x_1,\ldots,x_N)\ = \ c_N\left|\Delta(x_1^2,\ldots,x_N^2)\right|\prod_{j=1}^N |x_j| \prod_{j=1}^Ne^{-|x_j|^2/2}
\end{align}
where $\Delta$ denotes the Vandermonde determinant, and the complex Gaussian random variables have mean $0$ and variance $1$. Entries in matrices from the asymmetric ensemble are iidrv, while entries in the symmetric ensemble are iidrv in the upper triangle and the diagonal.

Note that the joint densities for singular values differ between the symmetric and asymmetric ensembles. In the ensembles that follow, we will chiefly consider symmetric matrices, and in doing so highlight the consistency found instead with the statistics we study for symmetric and asymmetric ensembles.

%%%%%%%%%%%%%%%%%%%%%%%%%%%%%%%%%%%%%%%%%%%%%%%%%%%%%%%%%%%%%%%%%%%%%
%%%%%%%%%%%%%%%%%%%%%%%%%%%%%%%%%%%%%%%%%%%%%%%%%%%%%%%%%%%%%%%%%%%%%
%%%%%%%%%%%%%%%%%%%%%%%%%%%%%%%%%%%%%%%%%%%%%%%%%%%%%%%%%%%%%%%%%%%%%
%%%%%%%%%%%%%%%%%%%%%%%%%%%%%%%%%%%%%%%%%%%%%%%%%%%%%%%%%%%%%%%%%%%%%
%%%%%%%%%%%%%%%%%%%%%%%%%%%%%%%%%%%%%%%%%%%%%%%%%%%%%%%%%%%%%%%%%%%%%
%%%%%%%%%%%%%%%%%%%%%%%%%%%%%%%%%%%%%%%%%%%%%%%%%%%%%%%%%%%%%%%%%%%%%
\subsubsection{Checkerboard Ensembles}
We investigate extensions of the structured ``checkerboard'' ensemble from \cite{RMT2016} into the complex regime. In that paper, the authors investigated a Hermitian ensemble, with real limiting eigenvalue distribution having almost all eigenvalues in a semicircular mass at the origin, referred to as the ``bulk'' and a vanishing percentage of eigenvalues, whose distribution is described explicitly, that moves off to infinity and is referred to as the ``blip.''(We adopt this terminology of bulk and blip where appropriate.)

The first complex analog we investigate is constructed as follows.

\begin{definition}
\label{def:complexSymCheck}
Fix $k\in\N$ and $w\in \C$. Then, for $k\mid N$, the $N\times N$ \textbf{complex symmetric $(k,w)$-checkerboard ensemble} is the ensemble of matrices $A$ with entries
\begin{align}
a_{ji}\ = \ a_{ij}\ =\
\begin{cases}
* & \text{ if } i\not\equiv j\pmod{k}
\\
w & \text{ if } i\equiv j\pmod{k}
\end{cases}
\end{align}
where $*\sim X+iY$ are selected such that $X$, $Y$ are iidrv mean $0$ variance $1/2$ real random variables. When we set $w=1$, or the value of $w$ is clear, we will just refer to the \textbf{complex symmetric $k$-checkerboard ensemble}.
\end{definition}

In contrast, the real symmetric ensemble studied in \cite{RMT2016} uses real random variables for $a_{ij}=a_{ji}$, and the Hermitian ensemble studied uses complex random variables with $a_{ij}=\overline{a_{ji}}$. In these situations, Hermiticity implies the resulting eigenvalue distributions are real. Our matrices are not necessarily Hermitian, and thus the eigenvalue distributions that arise are on $\C$. Restricting our attention to complex symmetric rather than the fully asymmetric case turns out to not make a difference for several of the following results. We have chosen to require symmetry, however, to highlight the difference between requiring symmetric structure in the real and complex settings (real symmetric and real asymmetric ensembles have very different behavior), and also to contrast with the differing behavior of complex symmetric and complex asymmetric Gaussian ensembles discussed above. For simplicity, we prove most of our results below for $w=1$ as was done in \cite{RMT2016} -- the extension to other values of $w$ is relatively straightforward.

In the paper \cite{RMT2016} studying the Hermitian version of this ensemble, the semicircular bulk was analyzed with the method of moments, but this could not be used for the blip as the eigenvalues were growing too rapidly. The blip existence was established by a perturbation argument using Weyl's inequalities (available for Hermitian matrices), and the distribution of the blip was analyzed using a polynomial weighting function.

None of these techniques are directly applicable for non-Hermitian ensembles with complex eigenvalue distributions. Complex polynomial weighting functions are not as well behaved -- for example, they are far from non-negative. Non-Hermitian ensembles also do not enjoy perturbation results such as Weyl's inequalities, as the spectra can be quite unstable due to the presence of pseudospectrum \cite{Tao2011}.

The method of moments also runs into serious difficulties in the complex regime. The use of the standard (real) method of moments is two-fold. Appropriate bounds on the moments implies convergence of the measures to a limiting measure (e.g. via the Carleman continuity theorem), and the moments also uniquely determine the limiting distribution. The analogous problem for complex moments uses mixed moments of the form
\begin{align}
\int z^{r_1}\overline{z}^{r_2}~d\mu.
\end{align}
However, these mixed moments do not have a straightforward relation to the matrix entries, as is available via the eigenvalue trace lemma in the real case and for moments of the form
\begin{align}
\int z^r~d\mu.
\end{align}
which we refer to as ``holomorphic.'' Although these holomorphic moments can be computed easily via the eigenvalue trace lemma for spectral measures, they cannot in general be used to characterize complex distributions. For example, all holomorphic moments of any angularly symmetric distribution will vanish. Ultimately, this is because the space of real polynomials is dense in various function spaces (the Stone-Weierstrass Theorem) and similarly for complex polynomials in $z$ and $\overline{z}$, but holomorphic polynomials in $z$ do not enjoy such properties \cite{Tao2011}.

Our analysis of the complex eigenvalues will thus employ markedly different techniques. As a proxy for the complex eigenvalues, we first study the associated singular value distributions, and explicitly describe the split limiting behavior in this context.

\begin{definition}
Given an $N\times N$ complex symmetric $k$-checkerboard matrix $A$, define the \textbf{bulk squared singular spectral measure} as
\begin{align}
\nu_{A,N}^{s^2}(x) \ = \ \frac1N \sum_{\sigma \text{ eigenvalue } A^*A}\delta\left(x-\frac{\sigma}{N}\right).
\end{align}
Note that $\sigma\geq0$ is a singular value of $A$ if and only if $\sigma^2$ is an eigenvalue of $B:=A^*A$.
\end{definition}

\begin{theorem}
\label{thm:singVal bulk}
Let $A_N$ be a random sequence of $N\times N$ complex symmetric $k$-checkerboard matrices. Then as $N\rightarrow\infty$, $\nu_{A_N,N}^{s^2}$ converges almost surely to the quarter-circular probability distribution (after renormalizing the total measure so that the distribution integrates to $1$) of radius $R=2\sqrt{1-1/k}$ and circle center at $0$, supported on $[0,R]$.
\end{theorem}
We also give an explicit description of the singular value blip distribution.

\begin{definition}
\label{def:EBSSSM}
The \textbf{empirical blip square singular spectral measure (EBSSSM)} for a matrix $A$ is

\begin{align}
\mu_{A,N}^{s^2}\ :=\ \frac{1}{k}\sum_{\sigma \text{ an eigenvalue of } A^*A} f_{n(N)}\left(\frac{k^2\sigma}{N^2}\right)\delta\left(x-\frac{1}{N}\left(\sigma-\frac{N^2}{k^2}\right)\right),
\end{align}
where $f_n(x)$ is the polynomial weighting function
\begin{align}
x^{2n}(x-2)^{2n}
\end{align}
and $n(N)$ is a monotonically growing function of $N$ that tends to $\infty$ such that $2^{4n(N)}=\littleO{N}$; for example, $n(N)=c\log N$ with $c$ a small enough constant suffices.
\end{definition}

Note that $\sigma$ an eigenvalue of $A^*A$ is equivalent to $\sqrt{\sigma}$ being a singular value of $A$. As in \cite{RMT2016}, the weight function $f$ weights the squared singular values in the blip roughly $1$, and weights the squared singular values in the bulk roughly $0$. The normalization factor $1/N$ ensures that we will find finite moments, i.e., the fluctuations of the squared singular values about the blip are of order $N$.

We can explicitly describe the blip distribution for the squared singular values, and recall a distribution studied in Theorem 1.9 of \cite{RMT2016}[Theorem~1.9], which also contains a few images of examples for small $k$.

\begin{definition}
\label{def:hollow GOE}
Fix $k\in\N$. Then the \textbf{$k\times k$ hollow Gaussian Orthogonal Ensemble (GOE)} is the ensemble of $k\times k$ matrices $A$ with entries
\begin{align}
a_{ji}\ =\ a_{ij}\ =\
\begin{cases}
* & \text{ if } i\neq j
\\
0 & \text{ if } i=j,
\end{cases}
\end{align}
where $*\sim X$ are iidrv mean $0$ variance $1$ real normal random variables, and the entries in the upper triangular half $A$ are all iidrv.
\end{definition}

When $k=2$, the empirical spectral measure is Gaussian, see \cite[Proposition~3.18]{RMT2016}. In general, standard universality implies that the limiting spectral distribution only requires the random variables to be mean $0$ and variance $1$. Furthermore, we use the term \emph{hollow} as a qualifier to any ensemble (for example, complex symmetric) where we have replaced the entries $a_{ij}$ with $0$ when $i\equiv j\pmod{k}$, with $k$ is clear from the context.

\begin{theorem}[Blip distribution for squared singular values]
\label{thm:singVal blip}
The empirical blip squared singular spectral measure of a complex symmetric $k$-checkerboard ensemble converges almost surely to the measure with $r$\textsuperscript{{\rm th}} centered moments equal to the $r$\textsuperscript{{\rm th}} centered moments of the empirical spectral measure of the $k\times k$ hollow Gaussian Orthogonal Ensemble, scaled by a factor of $(\sqrt{2}/k)^r$.
\end{theorem}

Note that this implies that the blip distribution of the squared singular values converges to the distribution of the hollow GOE scaled by $\sqrt{2}/k$. This is visualized in the $k=2$ case in Figure \ref{fig:singVal blip}.

\begin{figure}[h]
\scalebox{.8}{\includegraphics{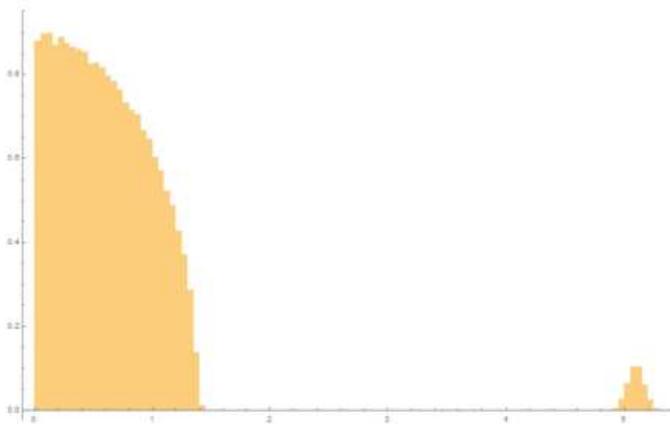}}
\caption{Normalized singular values of $100\times100$ complex symmetric $2$-checkerboard ensemble, $2000$ trials. Note the bulk and blip. This has not been re-scaled to display a quarter circle rather than a quarter ellipse.}
\label{fig:singVal blip}
\end{figure}

We also describe the bulk and blip behavior of the eigenvalues. In the preceding two results, analysis of the singular values was done via the method of moments, taking advantage of the Hermiticity associated with singular values.  %Similarly, the Hermiticity of the real symmetric checkerboard ensembles allowed \cite{RMT2016} to prove a semicircular law for the bulk with the method of moments.
%Non-Hermitian ensembles like our complex symmetric checkerboard, however, have required different techniques. Compared to Hermitian matrices, the spectra of non-Hermitian matrices are for example less stable under perturbation, and the moment method is ineffective for direct characteriztion of the eigenvalue distribution on $\C$ (see \cite{Tao2011}).
Since Girko in 1984 \cite{Girko1984}, however, work on such non-Hermitian ensembles has proceeded through the $\log$ potential and his Hermitization trick, with the limiting circular law distribution of fully random complex matrices being fully proven by Tao and Vu in 2010 \cite{TaoVu2010}. In analogy with the real method of moments, continuity of the log potential, closely related with the Stieltjes transform, plays a surrogate role to moment continuity theorems.

As short-hand, we refer to the ensemble with iidrv mean $0$ variance $1$ complex entries as the \emph{complex asymmetric ensemble}. Associated measures that arise below are denoted with a superscript ``asym.'' Similarly, measures that arise below in association with a complex symmetric checkerboard ensemble will be denoted with a superscript ``check.''

We also proceed with the $\log$ potential and Hermitization, and will show that, up to an explicit scaling factor and an assumption on the least singular values, our structured checkerboard ensembles also have a bulk that converges to a circular law.

\begin{theorem}[Eigenvalue Bulk - Complex Symmetric $k$-checkerboard]
\label{thm:complexSym bulk}
Consider a sequence of $N\times N$ random matrices $A_N$ from the complex symmetric $k$-checkerboard ensemble, with normalized spectral distribution $\mu_{\frac{1}{\sqrt{N}}A_N}$. Assume appropriate control of the least singular values as in Assumption \ref{assumption:singVal bounds}. Then, as $N\rightarrow\infty$, we have almost sure convergence $\mu_{\frac{1}{\sqrt{N}}A_N}\rightarrow\mu_R^\text{{\rm circ}}$ for $\mu_R^\text{{\rm circ}}$ the uniform measure on the disc centered at the origin with radius $R:=\sqrt{1-1/k}$.
\end{theorem}

See Figure \ref{fig:imagestoseelater} for a visualization of the bulk behavior in a more general setting. (The bulk corresponds to the large circular mass in the center.) This involves a careful combinatorial reduction that connects our complex symmetric checkerboard ensemble to the asymmetric case, via an interpretation of the Hermitianized moments as counting walks on certain trees.

%\begin{figure}[h]
%\scalebox{.3}{\includegraphics{Smiley.eps}}
%\caption{Numerical example of a distribution arising from a generalized $k$-checkerboard ensemble.}
%\label{fig:smileyIntro}
%\end{figure}

We also describe the position of the split-limiting eigenvalue blip, which will be naturally stated in the context of more general checkerboard ensembles.

\begin{definition}\label{def:gencheckone} We define a \emph{generalized $k$-checkerboard ensemble} to be an ensemble of matrices $A$ with entries either real/complex random variables or deterministic constants, that satisfy $a_{ij}=a_{mn}$ if $i\equiv m\pmod{k}$ and $j\equiv n\pmod{k}$, and such that for fixed $i,j$, $a_{ij}$ is always ``equal'' over all matrices in the ensemble (``equal'' in the sense that the entry in that position is always either the same deterministic value or random variable). The qualifiers symmetric/asymmetric refer to the structure we place on both the random variables and the deterministic entries, and real/complex refer to the random variables used.
\end{definition}

Note that the complex symmetric $k$-checkerboard ensemble from Definition \ref{def:complexSymCheck} is an example of a generalized $k$-checkerboard ensemble, where the deterministic entries are all $1$ and we set $a_{ij}=1$ when $i\equiv j\pmod{k}$. Indeed, many of the above results hold in this more general context as well.
\begin{example} This depicts a generalized $3$-checkerboard asymmetric ensemble, when the entries $*$ are iidrv complex random variables and the $w_i$ are fixed in value and position over the ensemble:
$$
\begin{pmatrix}
w_1 & * & w_2 & w_1 & * & w_2 \\
* & * & * & * & * & * & \cdots\\
* & w_3 & * & * & w_3 & * \\
w_1 & * & w_2 & w_1 & * & w_2 \\
* & * & * & * & * & * & \cdots\\
* & w_3 & * & * & w_3 & * \\
 & \vdots & & & \vdots &
\end{pmatrix}.
$$

\label{ex:genCheck}
\end{example}

\begin{definition}\label{def:genchecktwo} A generalized $k$-checkerboard ensemble is said to be \emph{$m$-regular} if, for any $N\equiv 0\pmod{k}$, there are $Nm/k$ deterministic entries in every row of all $N\times N$ matrices in the ensemble.
\end{definition}

For example, the complex symmetric $k$-checkerboard ensemble from Definition \ref{def:complexSymCheck} is $1$-regular, while the ensemble described in Example \ref{ex:genCheck} is not $m$-regular for any $m$. In this scenario, we find that the bulk results for singular values and the eigenvalues will also hold, up to scaling.

\begin{corollary}
\label{cor:gen singVal bulk}
Consider an $m$-regular generalized complex symmetric $k$-checkerboard ensemble. In analogy with Theorem \ref{thm:singVal bulk}, when $N\rightarrow\infty$ the squared singular values have moments
\begin{align}
M_r\ =\ \left(\frac{R}{2}\right)^{2r} C_r
\end{align}
for $C_r = \frac1{r+1}\ncr{2r}{r}$ the $r$\textsuperscript{{\rm th}} Catalan number and $R=2\sqrt{1-m/k}$, which shows that the bulk of the singular values converges almost surely to a quarter-circle distribution of radius $R$, with the circle's center at the origin.
\end{corollary}

\begin{corollary}
\label{cor:gen eigen bulk}
In analogy with Theorem \ref{thm:complexSym bulk}, consider a sequence of $N\times N$ random matrices $A_N$ from the $m$-regular complex symmetric $k$-checkerboard ensemble, with normalized spectral distribution $\mu_{\frac{1}{\sqrt{N}}A_N}$. Assume appropriate control of the least singular values in appropriate analogy to Assumption \ref{assumption:singVal bounds}. Then, as $N\rightarrow\infty$, we have almost sure convergence $\mu_{\frac{1}{\sqrt{N}}A_N}\rightarrow\mu_R^\text{{\rm circ}}$ for $\mu_R^\text{{\rm circ}}$ the uniform measure on the disc centered at the origin with radius $R:=\sqrt{1-m/k}$.
\end{corollary}

In the absence of Hermitian perturbation results, the characterization of a blip with different limiting behavior is not so readily obtainable for complex distributions. The techniques we use to characterize the complex eigenvalue blip will be markedly more involved than a short perturbation argument.

Fix a generalized $k$-checkerboard asymmetric ensemble.\footnote{We take the ensemble to be asymmetric instead of symmetric to accommodate general asymmetric patterns for the deterministic entries, see for example Example \ref{ex:genCheck}.} Note that any generalized $k$-checkerboard matrix $A$ can be decomposed as $A=M+P$ where $M$ is a generalized $k$-checkerboard matrix with all deterministic entries set to $0$, and $P$ is finite rank (at most $k$), completely deterministic, and composed of repeating blocks of some fixed $k\times k$ matrix $B$ (determined by the ensemble); we will use this notation when discussing the blip for generalized checkerboard matrices.

\begin{example}
\label{ex:Bmat}
For example, the $3\times 3$ matrix $B$ associated with the ensemble in Example \ref{ex:genCheck} is
$$B=\begin{pmatrix} w_1 & 0 & w_2 \\
0 & 0 & 0 \\
0 & w_3 & 0
\end{pmatrix}.$$
\end{example}
For $A_N$ an $N\times N$ matrix from the ensemble, we expect a vanishing proportion of the eigenvalues growing of order $N$ (referred to as the blip) and the remaining eigenvalues of size $N^{1/2}$ (referred to as the bulk) as this would correspond, heuristically, to the behavior of the singular values as in Proposition \ref{prop:singVal regimes}. One also expects, heuristically, that the spectral distribution should follow the distribution of the matrix $P$ up to an error of size $O(N^{1/2})$ from the matrix $M$, as occurs in the real case. Roughly speaking, we expect a clump of eigenvalues whose size is around the order of $N^{1/2}$ at each eigenvalue of $P$, with the bulk consisting of all the clumps associated to the zero eigenvalues of $P$, which has fixed rank at most $k$ as $N\rightarrow\infty$. The blip distribution, then, should reflect the distribution of the nonzero eigenvalues of $B$.

This heuristic seems to follow numerical simulation. See Figure \ref{fig:imagestoseelater} (left), which corresponds to an ensemble with matrix $B$ having eigenvalues chosen from roots of unity with appropriate multiplicity.

We give a justification for this heuristic and numerical understanding of the blip. To extract the blip position, we thus modify the empirical spectral measure $\mu_{A_N}$ using two types of renormalization -- dividing the matrix by $N$ so that the location of the blip is of constant order as $N\rightarrow\infty$, while the bulk is vanishing as $O(N^{-1/2})$, and multiplying the total measure by $N$ so that the measure of the blip remains constant rather than vanishing.

\begin{definition}
Let $A_N$ be an $N\times N$ matrix. Define the renormalized measure $\tilde{\mu}_{A_N}:=(N/k)\mu_{\frac{k}{N}A_N}$, where $\mu_{A_N}$ is the empirical spectral measure of $A_N$ (on $\C$).
\end{definition}

We wish to extract an almost sure limiting measure $\tilde{\mu}_{A_N}\rightarrow\tilde{\mu}$ as $N\rightarrow\infty$ over sequences of matrices $\{A_N\}$ from the ensemble. However, we expect such a measure $\tilde{\mu}$ to have a singularity at $0$, since each $\tilde{\mu}_{A_N}$ has total measure $N/k$, the bulk of which is of size $O(N^{-1/2})$, going to $0$ as $N\rightarrow\infty$.

To avoid this singularity, we will instead restrict our measures by excising small neighborhoods at the origin.

%We expect the eigenvalues in the blip to be a vanishing proportion of the $N$ eigenvalues of $A_N$. To extract their behavior, we need to scale and renormalize -- this introduces a singularity from the bulk near the origin, which we will need to excise.

\begin{notation*}
For $\epsilon>0$, let $\mathcal{B}_\epsilon=\{z:|z|\leq\epsilon\}\subset \C$ and $\Omega_\epsilon=\C\setminus \mathcal{B}_\epsilon$.
\end{notation*}

With some abuse of notation, we use $\tilde{\mu}_N$ to denote both the full measure on $\C$ and the measure restricted to $\Omega_\epsilon$ where appropriate. Instead of convergence of $\tilde{\mu}_{A_N}\rightarrow\tilde{\mu}$ on $\C$, we restrict to $\Omega_{\epsilon}$ to avoid the limiting singularity at $0$.

Unfortunately, even the existence of a limiting measure associated to appropriate normalized measures extracting blip behavior is not clear -- one might hope to proceed through the log potential, though certain normalization conditions will yield singularities that present serious obstacles. We show that, assuming a limiting measure exists, the limiting measure must indeed be characterized by the spectral distribution of $B$.

\begin{theorem}
\label{thm:complex blip}
Assume, restricted to $\Omega_\epsilon$, that $\tilde{\mu}_N\rightarrow\tilde{\mu}$ almost surely for every $\epsilon>0$. Then for any fixed $\epsilon>0$ smaller than all eigenvalues of $B$, $\tilde{\mu}$ must be the spectral measure of $B$ restricted to $\Omega_\epsilon$.
\end{theorem}

For example, with ensemble as in Example \ref{ex:genCheck}, this theorem states that the blip is described by the measure $\tilde{\mu}$ which will be the restriction to $\Omega_\epsilon$ of the spectral measure of the $3\times3$ matrix listed in Example \ref{ex:Bmat}.

\begin{remark}
In the theorem statement, we have neglected distinguishing $\tilde{\mu}$ restricted to $\Omega_\epsilon$ for different $\epsilon$, since $\tilde{\mu}$ on $\Omega_\epsilon$ restricts to $\tilde{\mu}$ on $\Omega_{\epsilon'}$ when $0<\epsilon'<\epsilon$.
\end{remark}

The basic idea is to show first that the limiting measure must be discrete and finitely supported on the nonzero eigenvalues of $B$, and to then show that holomorphic moments (calculated from the eigenvalue trace lemma) are enough to characterize discrete distributions, while also controlling the error from computing moments on $\Omega_\epsilon$ instead of all of $\C$.

As a corollary, this gives us better control on the total measure of the bulk, in analogy with the case of real eigenvalues.

\begin{corollary}
\label{cor:bulk mass}
Write $k'$ for the number of nonzero eigenvalues of $B$, with multiplicity. The bulk of the spectral measure $\mu_{A_N}$ consists of $N-k'$ eigenvalues of order $N^{1/2+\delta}$ for any $\delta>0$. That is, $\tilde{\mu}_{A_N}$ almost surely has total measure $N-k'$ on $\mathcal{B}_{N^{-1/2+\delta}}$ as $N\rightarrow\infty$.
\end{corollary}
%--------------------------------%

Sections \ref{section:complexCheck} and \ref{section:genCheck} give proofs for these results. In Section \ref{section:complexCheck}, we prove Theorem \ref{thm:singVal bulk} and Theorem \ref{thm:singVal blip}, the bulk and blip results of the singular values for complex symmetric checkerboard ensembles, as well as Theorem \ref{thm:complexSym bulk}, our bulk result for the complex eigenvalue distribution of complex symmetric checkerboard ensembles. In Section \ref{section:genCheck}, we prove singular value and eigenvalue bulk analogs in Corollary \ref{cor:gen singVal bulk} and Corollary \ref{cor:gen eigen bulk} for generalized checkerboard matrices, and prove Theorem \ref{thm:complex blip} and Corollary \ref{cor:bulk mass} to describe the complex blip behavior. We conclude with some conjectural observations concerning generalized checkerboard matrices and related ensembles in Subsection \ref{section:conjectures}. Some terminology and auxiliary material can be found in Appendix \ref{appendix:terminology}.

%%%%%%%%%%%%%%%%%%%%%%%%%%%%%%%%%%%%%%%%%%%%%%%%%%%%%%%%%%%%%%%%%%%%%%%%%%%%%%%%%%%%%%%%%%%%%%%%%%%%%%%%%%%%%%%%
%%%%%%%%%%%%%%%%%%%%%%%%%%%%%%%%%%%%%%%%%%%%%%%%%%%%%%%%%%%%%%%%%%%%%%%%%%%%%%%%%%%%%%%%%%%%%%%%%%%%%%%%%%%%%%%%
%%%%%%%%%%%%%%%%%%%%%%%%%%%%%%%%%%%%%%%%%%%%%%%%%%%%%%%%%%%%%%%%%%%%%%%%%%%%%%%%%%%%%%%%%%%%%%%%%%%%%%%%%%%%%%%%
%%%%%%%%%%%%%%%%%%%%%%%%%%%%%%%%%%%%%%%%%%%%%%%%%%%%%%%%%%%%%%%%%%%%%%%%%%%%%%%%%%%%%%%%%%%%%%%%%%%%%%%%%%%%%%%%

%--------------------------------%
%%%%%%%%%%%%%%%%%%%%%%%%%%%%%%%%%%%%%%%%%%%%%%%%%%%%%%%%%%%%%%%%%%%%%%%%%%%%%%%%%%%%%%%%%%%%%%%%%%%%%%%%%%%%%%%%
%%%%%%%%%%%%%%%%%%%%%%%%%%%%%%%%%%%%%%%%%%%%%%%%%%%%%%%%%%%%%%%%%%%%%%%%%%%%%%%%%%%%%%%%%%%%%%%%%%%%%%%%%%%%%%%%
%%%%%%%%%%%%%%%%%%%%%%%%%%%%%%%%%%%%%%%%%%%%%%%%%%%%%%%%%%%%%%%%%%%%%%%%%%%%%%%%%%%%%%%%%%%%%%%%%%%%%%%%%%%%%%%%
%%%%%%%%%%%%%%%%%%%%%%%%%%%%%%%%%%%%%%%%%%%%%%%%%%%%%%%%%%%%%%%%%%%%%%%%%%%%%%%%%%%%%%%%%%%%%%%%%%%%%%%%%%%%%%%%
\section{Complex checkerboard ensembles}
\label{section:complexCheck}

We first establish the existence of two squared singular value regimes with a matrix perturbation result.

\begin{proposition}
\label{prop:singVal regimes}
As $N\rightarrow\infty$, the squared singular values of $k$-checkerboard complex symmetric matrices almost surely fall into two regimes: $N-k$ of the squared singular values are $\bigO{N^{1+\epsilon}}$, and $k$ of the squared singular values are $N^2/k^2+\bigO{N^{3/2+\epsilon}}$, for any $\epsilon>0$.
\end{proposition}

\begin{proof}
A $k$-checkerboard matrix $A$ can be decomposed as $M+P$, where
\begin{align}
m_{i,j}\ =\
\begin{cases}
a_{i,j} & \text{if } i\not\equiv j \pmod{k}
\\
0 & \text{otherwise}
\end{cases}
\hspace{2 cm}
p_{i,j}\ =\
\begin{cases}
0 & \text{if } i\not\equiv j \pmod{k}
\\
a_{i,j} & \text{otherwise.}
\end{cases}
\end{align}

A straightforward generalization of \cite[Lemma~B.3]{RMT2016} in the context of our above argument for the square singular values bulk shows that as $N\rightarrow\infty$, $\norm{A}_{\mathrm{op}}=\bigO{N^{1/2+\epsilon}}$ almost surely. Since $P$ has $k$ singular values at $N/k$, and $N-k$ eigenvalues at $0$, Weyl's inequality for singular values implies that almost surely, $N-k$ of the singular values are $\bigO{N^{1/2+\epsilon}}$, and $k$ of the singular values are $N/k+\bigO{N^{1/2+\epsilon}}$. This implies the proposition for squared singular values.
\end{proof}

We modify the combinatorics and weighting function from \cite{RMT2016} to extract the limiting distribution of the blip for the squared singular values of complex symmetric checkerboard matrices.

%\begin{definition}
%\label{def:EBSSSM}
%The \textbf{empirical blip square singular spectral measure (EBSSSM)} for a matrix $A$ is

%\begin{align}
%\mu_{A,N}^{s^2}\ :=\ \frac{1}{k}\sum_{\sigma \text{ an eigenvalue of } A^*A} f_{n(N)}\left(\frac{k^2\sigma}{N^2}\right)\delta\left(x-\frac{1}{N}\left(\sigma-\frac{N^2}{k^2}\right)\right),
%\end{align}
%where $f_n(x)$ is the polynomial weighting function
%\begin{align}
%x^{2n}(x-2)^{2n}
%\end{align}
%and $n(N)$ is a monotonically growing function of $N$ that tends to $\infty$ such that $2^{4n(N)}=\littleO{N}$; for example, $n(N)=c\log N$ with $c$ a small enough constant suffices.
%\end{definition}

%%%%%%%%%%%%%%%%%%%%%%%%%%%%%%%%%%%%%%%%%%%%%%%%%%%%%%%%%%%%%%%%%%%%%%%%%%%%%%%%%%%%%%%%%%%%%%%%%%%%%%%%%%%%%%
%%%%%%%%%%%%%%%%%%%%%%%%%%%%%%%%%%%%%%%%%%%%%%%%%%%%%%%%%%%%%%%%%%%%%%%%%%%%%%%%%%%%%%%%%%%%%%%%%%%%%%%%%%%%%%
%%%%%%%%%%%%%%%%%%%%%%%%%%%%%%%%%%%%%%%%%%%%%%%%%%%%%%%%%%%%%%%%%%%%%%%%%%%%%%%%%%%%%%%%%%%%%%%%%%%%%%%%%%%%%%
\subsection{Singular values of complex checkerboard matrices: bulk}\ \\

In this subsection we establish the limiting bulk measure for singular values of complex symmetric $k$-checkerboard matrices.

We use the method of moments. We wish to match the moments of our limiting squared singular value distribution with the moments of the quarter-circular distribution.

\begin{proposition}
\label{prop:quartercircle moments}
Let $X$ be the random variable with probability density function a quarter circle supported on $[0,2]$ of radius $2$ with circle-center at the origin (normalized by $\pi$ so that it is a probability distribution). Then the random variable $X^2$ has its $r$\textsuperscript{{\rm th}} moment, $M_r$, equal to the $r$\textsuperscript{{\rm th}} Catalan number
\begin{align}
C_r\ :=\ \frac{1}{r+1}\binom{2r}{r}.
\end{align}
\end{proposition}

\begin{proof}
The even $2r$\textsuperscript{{\rm th}} moments of both the semicircular distribution and the quarter-circular distribution of $X$ are known to equal the $r$\textsuperscript{{\rm th}} Catalan number, see for example \cite{BaiSilverstein2010}. The proposition then follows by noting that the $r$\textsuperscript{{\rm th}} moment of the random variable $X^2$ is also the $2r$\textsuperscript{{\rm th}} moment of random variable $X$.
\end{proof}

\begin{proof}[Proof of Theorem \ref{thm:singVal bulk}]
Before we can apply the method of moments, we must first consider the perturbation as in the proof of Proposition \ref{prop:singVal regimes}, which exhibits complex symmetric $k$-checkerboard matrices as a finite rank (i.e., fixed as $N\rightarrow\infty$) perturbation from the corresponding hollow complex symmetric $k$-checkerboard matrices. Then, since $M+P$ a finite rank perturbation of $M$ implies that $(M+P)^*(M+P)$ is a finite rank perturbation of $M^*M$, we can apply Theorem 1.3 of \cite{RMT2016} to find that the complex symmetric $k$-checkerboard ensemble and the corresponding hollow ensemble have the same limiting squared singular value distribution. We thus apply the method of moments below to the hollow ensemble.

By the eigenvalue trace lemma and linearity of expectation, the $r$\textsuperscript{{\rm th}} moment of the bulk squared singular spectral measure is computed as
\begin{align}
\E\left[\nu_{A,N}^{s^2(r)}\right] & \ = \ \E\left[\int  \nu_{A,N}^{s^2}(x)x^r\;dx\right] \ = \ \frac1N \ \E\left[\sum_{\sigma \text{ eigenvalue } B}\left(\frac{\sigma}{N}\right)^r\right]\nonumber\\
& \ = \ N^{-r-1}\E\left[\sum_{\sigma \text{ eigenvalue } B}\sigma^r\right] \ = \ N^{-r-1}\E\left[\mathrm{Tr}(B^k)\right]\nonumber\\
& \ = \ N^{-r-1}\sum_{1\le i_1,\cdots, i_r\le N}\E[b_{i_1,i_2}\cdots b_{i_r,i_1}],
\end{align}
where the entries of $B=A^*A$ are given by
\begin{align}
b_{ij} \ = \ \sum_{1\le k\le N} \overline{a_{ki}} a_{kj} \ = \ \sum_{1\le k\le N} \overline{a_{ik}} a_{kj}
\end{align}
by symmetry. Hence
\begin{align}
\E\left[\nu_{A,N}^{s^2(r)}\right] & \ = \ N^{-r-1}\sum_{1\le i_1\le\cdots\le i_r}\E[b_{i_1 i_2}\cdots b_{i_r i_1}]\nonumber\\
& \ = \ N^{-r-1}\sum_{1\le i_1,\cdots, i_r\le N}\sum_{1\le k_1,\ldots,k_r\le N}\E[\overline{a_{i_1k_1}}a_{k_1i_2}\cdots\overline{a_{i_rk_r}}a_{k_ri_1}]\nonumber\\
& \ = \ N^{-r-1}\sum_{1\le i_1,\cdots, i_{2r}\le N} \E[\overline{a_{i_1i_2}}a_{i_2i_3}\cdots\overline{a_{i_{2r-1}i_{2r}}}a_{i_{2r}i_1}]. \label{eq:singVal moments}
\end{align}
Each term $\zeta_I = \overline{a_{i_1i_2}}a_{i_2i_3}\cdots\overline{a_{i_{2r-1}i_{2r}}}a_{i_{2r}i_1}$ in the sum corresponds to a cyclic sequence $I = i_1\cdots i_{2r}$. Then $I$ may be associated to a closed walk on the complete graph with vertices labeled by the elements of the set $\{i_1, ..., i_{2r}\}$ in the order that the vertices are visited. Define the \emph{weight} of $I$ to be the number of distinct entries of $I$. If the weight of $I$ is at least $r+2$, then there is a factor $a$ in $\zeta_I$ independent from all the rest, and thus the expectation $\E[\zeta_I]=0$ (recall we have zeroed out all deterministic entries because we are considering the \emph{hollow} ensemble, and the random variables are all mean $0$).

The sequences of weight at most $r$ contribute negligibly, $o(N^{r+1})$. This is because the sequences may be partitioned into a finite number of equivalence classes by the isomorphism class of the corresponding walk. An isomorphism class of weight $t\le r$ then gives rise to $O(N^t)$ walks of weight $t$ by choosing labels for the distinct nodes in any such walk.

Closer analysis is required for a sequence $I$ of weight $r+1$. First, the walk corresponding to $I$ visits $r + 1$ distinct nodes and traverses $r$ distinct edges. Hence the walk consists of $2r$ steps on a tree with $r+1$ nodes.

Note that $\E[\zeta_I]$ contributes to the sum precisely when all the factors, $a_{ij}$, are matched with their conjugates, $\overline{a_{ij}}$, in which case $\E[\zeta_I]=1$. Indeed, if $a_{ij}$ is an entry of a complex  symmetric $k$-checkerboard matrix $A$ with $i\not\equiv j\pmod{k}$, note that $\E\left[a_{ij}a_{ij}\right]=\E\left[\overline{a_{ij}}\overline{a_{ij}}\right]=0$ while $\E\left[\overline{a_{ij}}a_{ij}\right]=1$. This is because if $a_{ij}\sim X+iY$ for $X,Y$ iidrv mean zero variance $1/2$ random variables, then
\begin{align}
\label{eq:complex expectation1}
\E\left[a_{ij}a_{ij}\right]&\ =\ \E\left[X^2\right]+2i\E\left[X\right]\E\left[Y\right]-\E\left[Y^2\right]\ =\ 0
\\
\label{eq:complex expectation2}
\E\left[\overline{a_{ij}}a_{ij}\right]&\ =\ \E\left[X^2\right]+\E\left[Y^2\right]\ =\ 1.
\end{align}
Thus it suffices to count the number of sequences $I$ satisfying the above condition. In the graph correspondence, the condition on $I$ is equivalent to the walk traversing each tree edge exactly twice, where, for an edge corresponding to $\{i,j\}$ in $I$, one traversal corresponds to $a_{i,j}$ and the other traversal to $\overline{a_{ij}}$ in $\zeta_I$. For a given edge $e$ and the corresponding subwalk $w$ between the first and second traversal of $e$, each edge in $w$ must be traversed and later retraced in the reverse direction, since trees are acyclic. Thus $w$ has an even number of steps, so the two traversals of $e$ occur on steps of opposite parity.
\begin{remark} This implies that the same result holds for the corresponding asymmetric ensemble, since this parity requirement ensures that the combinatorics must be the same in both cases.
\end{remark}
This corresponds in $\zeta_I$ to matching $a_{ij}$  with $\overline{a_{ij}}$; if the steps occurred with the same parity, then $a_{ij}$ would be matched with $a_{ij}$ (or $\overline{a_{ij}}$ with $\overline{a_{ij}}$), resulting in zero expectation. In summary, it suffices to count the number of non-isomorphic trees on $r+1$ nodes with a given starting node, and a given absolute order on the leaves---there is a bijection between such walks and such \emph{ordered} trees given by the order in which the leaves are visited in the walk.

As is well known, there are $C_r$ ordered trees on $r + 1$ nodes, where $C_r$ is the $r$\textsuperscript{{\rm th}} Catalan number \cite{Stan1999}. There is a further restriction: since $a_{ij}=0$ if $i\equiv j\pmod{k}$, an appearance of any such term in the cyclic product will contribute zero expectation. We may then label the nodes in the tree in such a way that no two adjacent nodes have the same congruence in $N^{r+1}\left(\frac{k-1}{k}\right)^r + o(N^{r+1})$ ways. Thus we have
\begin{align}
\E\left[\nu_{A,N}^{s^2(r)}\right] & \ = \ N^{-r-1} \ \left( \sum_{\textrm{weight }I < r+1} + \sum_{\textrm{weight }I = r+1} + \sum_{\textrm{weight }I > r+1} \right)\E[\zeta_I]\nonumber\\
& \ = \ N^{-r-1}\left(o(N^{r+1}) + \textrm{C}_r\left(N^{r+1}\left(\frac{k-1}{k}\right)^r + o(N^{r+1})\right) +0\right)\nonumber\\
& \ = \ \textrm{C}_r\left(\frac{k-1}{k}\right)^r + o(1).
\end{align}
Hence we have proved
\begin{align}
\lim_{N\to\infty} \E\left[\nu_{A,N}^{s^2(r)}\right]\ =\ \textrm{C}_r\left(\frac{k-1}{k}\right)^r\ =\ \textrm{C}_r\left(\frac{R}{2}\right)^{2r},
\end{align}
which are the moments of the (square of the) quarter circle distribution of radius $R = 2\sqrt{1-1/k}$ as in Proposition \ref{prop:quartercircle moments}, which suffices.
\end{proof}

%%
%%
%Note that (via Wolfram Alpha)
%$$\int_0^1 u^r\sqrt{1-u^2}\;du = \frac{\sqrt{\pi}}{4}\frac{\Gamma(\frac{r+1}2)}{\Gamma(\frac{r}2+2)}$$
%so the $r$th moment of the quarter circle (of radius $R$) is
%\begin{align}
%M_r=\int_0^R x^r\sqrt{R^2-x^2}\;dx &= \int_0^1 (Ru)^r\sqrt{R^2-(Ru)^2}(R\;du)\\
%& = R^{2+r}\int_0^1 u^r\sqrt{1-u^2}\;du\\
%& = R^{2+r}\frac{\sqrt{\pi}}{4}\frac{\Gamma(\frac{r+1}2)}{\Gamma(r+\frac12)}
%\end{align}
%using $x=Ru$. If $r=2l+1$ is odd, then
%\begin{align}
%M_r = R^{2+r}\frac{\sqrt{\pi}}{4}\frac{l!}{\Gamma(\frac{r}2+2)}.
%\end{align}
%Since $\Gamma(1/2)=\sqrt\pi$, we have
%\begin{align}
%\Gamma(\frac{r}2+2)&=\Gamma(l+2+\frac12)\\
%&=(l+1+1/2)(l-1/2)(l-3/2)\cdots(1/2)\Gamma(1/2)\\
%& = \sqrt\pi \frac{(2l+3)(2l+1)(2l-1)\cdots1}{2^{l+2}}\\
%& = \sqrt\pi \frac{(2l+3)!}{2^{2l+3}(l+1)!}\\
%\end{align}
%and thus
%\begin{align}
%M_r & = R^{2+r}2^{2l+3}\frac{l!(l+1)!}{4(2l+3)!}\\
%& = \frac{(2R)^{r+2}}{4(l+1)(l+2)C_{l+1}}
%\end{align}

%%%%%%%%%%%%%%%%%%%%%%%%%%%%%%%%%%%%%%%%%%%%%%%%%%%%%%%%%%%%%%%%%%%%%%%%%%%%%%%%%%%%%%%%%%%%%%%%%%%%%%%%%%%%%%
%%%%%%%%%%%%%%%%%%%%%%%%%%%%%%%%%%%%%%%%%%%%%%%%%%%%%%%%%%%%%%%%%%%%%%%%%%%%%%%%%%%%%%%%%%%%%%%%%%%%%%%%%%%%%%
%%%%%%%%%%%%%%%%%%%%%%%%%%%%%%%%%%%%%%%%%%%%%%%%%%%%%%%%%%%%%%%%%%%%%%%%%%%%%%%%%%%%%%%%%%%%%%%%%%%%%%%%%%%%%%
\subsection{Singular values of complex checkerboard matrices: blip}\ \\

Throughout this entire subsection, we follow the notation and terminology in \cite{RMT2016}. For convenience, the relevant terminology from that paper is collected in Appendix \ref{appendix:terminology}. We analyze the blip using the method of moments.

\begin{lemma}
\label{lemma:rthMomentEBSSSM}
The expected $r$\textsuperscript{{\rm th}} moment of the EBSSSM is given by
\begin{align}
\E\left[M_{A,N}^{s^2(r)}\right]\ =\ \frac{1}{k^{r+1}}\sum_{j=0}^{2n}\binom{2n}{j}\sum_{i=0}^{r+j} \binom{r+j}{i}(-1)^{r-i}\left(\frac{k}{N}\right)^{4n+2i+r}\E\left[\mathrm{Tr}(A^*A)^{2n+i}\right].
\end{align}
\end{lemma}

\begin{proof} By definition
\begin{align}
\E\left[M_{A,N}^{s^2(r)}\right]&\ =\ \frac{1}{k}\E\left[\sum_\sigma f_n\left(\frac{k^2\sigma}{N^2}\right)\left(\frac{1}{N}\left(\sigma-\frac{N^2}{k^2}\right)\right)^r\right] \nonumber
\\
&\ =\ \frac{1}{k}N^{-r}\E\left[\sum_\sigma f_n\left(\frac{k^2\sigma}{N^2}\right)\left(\sigma-\frac{N^2}{k^2}\right)^r\right] \nonumber
\\
&\ =\ \frac{1}{k}N^{-r}\E\left[\sum_\sigma \left(\frac{k^2\sigma}{N^2}\right)^{2n}\left(\frac{k^2\sigma}{N^2}-1-1\right)^{2n}\left(\sigma-\frac{N^2}{k^2}\right)^r\right] \nonumber
\\
&\ =\ \frac{1}{k}N^{-r}\left(\frac{k^2}{N^2}\right)^{2n} \E\left[\sum_\sigma \sigma^{2n} \sum_{j=0}^{2n} \binom{2n}{2n-j}(-1)^{2n-j}\left(\frac{k^2\sigma}{N^2}-1\right)^j\left(\sigma-\frac{N^2}{k^2}\right)^r\right]\nonumber
\\
&\ =\ \frac{1}{k}N^{-r}\left(\frac{k^2}{N^2}\right)^{2n}\sum_{j=0}^{2n}\binom{2n}{j}\sum_{i=0}^{r+j} \binom{r+j}{i}\left(-\frac{N^2}{k^2}\right)^{r-i}\E\left[\sum_\sigma \sigma^{2n+i}\right] \nonumber
\\
&\ =\ \frac{1}{k^{r+1}}\sum_{j=0}^{2n}\binom{2n}{j}\sum_{i=0}^{r+j} \binom{r+j}{i}(-1)^{r-i}\left(\frac{k}{N}\right)^{4n+2i-r}\E\left[\mathrm{Tr}(A^*A)^{2n+i}\right],
\end{align}
where the last equality follows by the eigenvalue trace lemma.
\end{proof}

Observe that the $(i,j)$\textsuperscript{{\rm th}} entry of $A^*A$ is given by $\sum_{m=1}^N \overline{a_{m i}}a_{m j}=\sum_{m=1}^N \overline{a_{i m}}a_{m j}$ (using the symmetry condition of $A$). By definition of the trace
\begin{align}
\E\left[\mathrm{Tr}(A^*A)^\eta\right]\ =\ \sum_{\substack{1\leq i_1,\ldots,i_\eta\leq N \\ 1\leq m_1,\ldots m_\eta\leq N}} \E\left[\overline{a_{i_1 k_1}}a_{m_1 i_2}\overline{a_{i_2 m_2}}a_{k_2 i_3}\cdots \overline{a_{i_\eta m_\eta}}a_{m_\eta i_1}\right].
\end{align}
Terms of the form $$\E\left[\mathrm{Tr}(A^*A)^\eta\right]=\sum_{\substack{1\leq i_1,\ldots,i_\eta\leq N \\ 1\leq m_1,\ldots m_\eta\leq N}} \E\left[\overline{a_{i_1 m_1}}a_{m_1 i_2}\overline{a_{i_2 m_2}}a_{m_2 i_3}\cdots \overline{a_{i_\eta m_\eta}}a_{m_\eta i_1}\right]$$ will be our cyclic products. Degrees of freedom arguments allow us to restrict our attention to ``configurations'' of ``$1$-blocks'' and ``$2$-blocks.'' See Appendix \ref{appendix:terminology} for terminology taken from \cite{RMT2016} and \cite[Lemma~3.13]{RMT2016} for proof of the claim. We compute the contribution to the expectation $\E\left[\mathrm{Tr}(A^*A)^\eta\right]$.

\begin{lemma}
\label{lemma:classContribution}
The total contribution to $\E\left[\mathrm{Tr}(A^*A)^\eta\right]$ of an $S$-class $C$ with $r_1$ $1$-blocks and $(|S|-r_1)$ $2$-blocks is
\begin{align}
p(\eta)\binom{|S|}{r_1}(k-1)^{|S|-r_1}\left(\frac{1}{2}\right)^{r_1/2}\E_k\mathrm{Tr}B^{r_1}\left(\left(\frac{N}{k}\right)^{2\eta-|S|}
+O\left(\left(\frac{N}{k}^{2\eta-|S|-1}\right)\right)\right),
\end{align}
where
\begin{align}
p(\eta)\ =\ \frac{(2\eta)^{|S|}}{|S|!}+O(\eta^{|S|-1})
\end{align}
and the expectation $\E\left[\mathrm{Tr}B_k^{m_1}\right]$ is taken over the $k\times k$ hollow GOE as defined in Definition \ref{def:hollow GOE}.
\end{lemma}

\begin{proof}
The quantity $p(\eta)$ expresses the number of ways to set the position of $|S|$ blocks (which we have established must be $1$-blocks or $2$-blocks) among a cyclic product of length $2\eta$ which arises from $\E\left[\mathrm{Tr}(A^*A)^\eta\right]$. We can estimate $p(\eta)=\binom{2\eta}{|S|}+\bigO{\eta^{|S|-1}}$. The term $\binom{2\eta}{|S|}$ counts the number of ways to choose positions of the blocks (ignoring overlap). The error term $\bigO{\eta^{|S|-1}}$ counts the number of ways in which some two blocks will be less than one term apart, which will occur non-generically as $\eta\rightarrow\infty$.

Next, the term $\binom{|S|}{r_1}$ counts the number of ways to choose which of the $|S|$ blocks are $1$-blocks (equivalently, the number of ways to choose which of the $|S|$ blocks are $2$-blocks). As in \cite[Proposition~3.14]{RMT2016}, the congruence classes modulo $k$ of all the indices is completely determined by the choices of congruence class for the indices of the $r_1$ $1$-blocks, and the following $(k-1)^{|S|-r_1}$ choices of congruence class for the shared index of each $2$-block. The $r_1$ $1$-blocks form a cyclic product of length $r_1$, and the number of ways of choosing the congruence classes modulo $k$ of their indices is equivalent to an expectation of the form
\begin{align}
\sum_{1\leq i_1,\ldots,i_{r_1}\leq k} \E\left[b_{i_1 i_2}b_{i_2 i_3}\cdots b_{i_{r_1}i_1}\right]\ =\ \E\left[\mathrm{Tr}B_k^{r_1}\right]
\end{align}
for $B_k$ as defined prior to Theorem \ref{thm:singVal blip}. This is because the number of valid choices of index congruence classes corresponds precisely to the number of ways to match terms in a length $r_1$ cyclic product, with the restriction that consecutive indices cannot be equal (which would correspond to a deterministic entry in the original checkerboard matrix, and not a $1$-block type entry). Further details in the argument for this reduction are similar to those in the proof of \cite[Proposition~3.14]{RMT2016}.

However, our extension to singular values requires a modification to the combinatorics. As in \eqref{eq:complex expectation1} and \eqref{eq:complex expectation2}, we see that the paired entries in our cyclic product must be matched in conjugate pairs if they are to contribute to the expectation. This is automatically the case for every $2$-block, since the two terms side by side are already conjugate pairs. However, as $\eta\rightarrow\infty$ and $|S|$ remains fixed, this will be true with probability $1/2$ for each pair of $1$-blocks, and since there are $r_1/2$ pairs of $1$-blocks, we see that the number of valid configurations should be scaled by $\left(\frac{1}{2}\right)^{r_1/2}$.

The last piece of the expression in Lemma \ref{lemma:classContribution} is the term $\left(\left(\frac{N}{k}\right)^{2\eta-|S|}+O\left(\left(\frac{N}{k}^{2\eta-|S|-1}\right)\right)\right)$, which arises from the degree of freedom count $2\eta-|S|-1$ for the indices once we have fixed their congruence classes, and the big $O$ error term arises from the lower degree of freedom terms we are ignoring, when only considering configurations of $1$-blocks and $2$-blocks.
\end{proof}

\begin{proof}[Proof of Theorem \ref{thm:singVal blip}]
We can now compute the $r$\textsuperscript{{\rm th}} moment $\E\left[M_{A,N}^{s^2(r)}\right]$ of the EBSSSM. By a combinatorial lemma \cite[Lemma~3.16]{RMT2016} copied as Lemma \ref{lemma:combinatorial} in Appendix \ref{appendix:terminology}, we see that only $S$-classes of size $r$ contribute: if $|S|>r$, then the contribution vanishes in the limit of $N\rightarrow\infty$ by a degree of freedom count, and if $|S|<r$, the contribution cancels via the combinatorial lemma. Then the outer sum in Lemma \ref{lemma:rthMomentEBSSSM} collapses to only the $j=0$ term. We can substitute Lemma \ref{lemma:classContribution} into Lemma \ref{lemma:rthMomentEBSSSM} (after adding a sum over the parameter $r_1$ which counts the number of $1$-blocks in our $S$-class):
\begin{align}
\lim_{N\rightarrow\infty} \E\left[M_{A,N}^{s^2 (r)}\right]&\ =\ \frac{1}{k^{r+1}}\sum_{i=0}^r\binom{r}{i}(-1)^{r-i}\sum_{r_1=0}^r\frac{(4n+2i)^r}{r!}\binom{r}{r_1}(k-1)^{r-r_1}\E_k\left[\mathrm{Tr}\left(\frac{1}{\sqrt{2}}B\right)^{r_1}\right] \nonumber
\\
&\ =\ \frac{1}{k^{r+1}}2^r\sum_{r_1=0}^r\binom{r}{r_1}(k-1)^{r-r_1}\E_k\left[\mathrm{Tr}\left(\frac{1}{\sqrt{2}}B\right)^{r_1}\right].
\label{eq:r centered moms}
\end{align}
To compute the $r$\textsuperscript{{\rm th}} centered moments, we need the first moment:
\begin{align}
\lim_{N\rightarrow\infty}\E\left[M_{A,N}^{s^2 (1)}\right]&\ =\ \frac{1}{k^2}\sum_{i=0}^{1} (-1)^{1-i}\binom{1}{i}(4n+2i)k(k-1)\ =\ \frac{2(k-1)}{k}.
\end{align}
Thus the $r$\textsuperscript{{\rm th}} centered moments are given by
\begin{align}
M_c^{s^2(r)}
&\ =\ \lim_{N\rightarrow\infty}\E\left[\int\left(x-\mu_{A,N}^{s^2(1)}\right)^r d\mu_{A,N}^{s^2}\right] \nonumber
\\
&\ =\ \sum_{r_1=0}^r\binom{r}{r_1}\left(\frac{-2(k-1)}{k}\right)^{r-r_1}\lim_{N\rightarrow\infty}\E\left[\mu_{A,N}^{s^2(r_1)}\right].
\end{align}
Substituting the expression from Equation (\ref{eq:r centered moms}) gives
\begin{align}
M_c^{s^2(r)} \ =\ \sum_{r_1=0}^r\binom{r}{r_1}(-1)^{r-r_1}\left(\frac{2}{k}\right)^r\frac{1}{k}\sum_{i=0}^{r_1}\binom{r_1}{i}(k-1)^{r-i} \E_k\left[\mathrm{Tr}\left(\frac{1}{\sqrt{2}}B\right)^i\right]. \nonumber
\end{align}
Next, using the identity $\binom{r}{r_1}\binom{r_1}{i}=\binom{r}{i}\binom{r-i}{r_1-i}$, we obtain
\begin{align}
M_c^{s^2(r)}&\ =\ \sum_{i=0}^r \binom{r}{i}(k-1)^{r-i}\left(\frac{2}{k}\right)^r\frac{1}{k}\E_k\left[\mathrm{Tr}\left(\frac{1}{\sqrt{2}}B\right)^i\right]\sum_{r_1=i}^r\binom{r-i}{r_1-i}(-1)^{r-r_1} \nonumber
\\
&\ =\ \sum_{i=0}^r \binom{r}{i}(k-1)^{r-i}\left(\frac{2}{k}\right)^r\frac{1}{k}\E_k\left[\mathrm{Tr}\left(\frac{1}{\sqrt{2}}B\right)^i\right](-1)^r\delta_{mi} \nonumber
\\
&\ =\ \left(\frac{\sqrt{2}}{k}\right)^r\frac{1}{k}\E_k\left[\mathrm{Tr}(B)^r\right]
\end{align}
which proves Theorem \ref{thm:singVal blip} via the moment method.
\end{proof}

%%%%%%%%%%%%%%%%%%%%%%%%%%%%%%%%%%%%%%%%%%%%%%%%%%%%%%%%%%%%%%%%%%%%%%%%%%%%%%%%%%%%%%%%%%%%%%%%%%%%%%%%%%%%%%
%%%%%%%%%%%%%%%%%%%%%%%%%%%%%%%%%%%%%%%%%%%%%%%%%%%%%%%%%%%%%%%%%%%%%%%%%%%%%%%%%%%%%%%%%%%%%%%%%%%%%%%%%%%%%%
%%%%%%%%%%%%%%%%%%%%%%%%%%%%%%%%%%%%%%%%%%%%%%%%%%%%%%%%%%%%%%%%%%%%%%%%%%%%%%%%%%%%%%%%%%%%%%%%%%%%%%%%%%%%%%
\subsection{Eigenvalues of complex checkerboard matrices: bulk}\ \\

The standard Hermitization process via the $\log$ potential is done as follows (see for example \cite{Tao2011}). Given a sequence of $N\times N$ random matrices $A_N$ with normalized spectral distribution $\mu_{\frac{1}{\sqrt{N}}A_N}$ on $\C$, we have the logarithmic potential
\begin{align}
f_N(z)\ :=\ \int_{\C}\log\abs{w-z}d\mu_{\frac{1}{\sqrt{N}}A_N}(z).
\end{align}
The key tool is the logarithmic potential continuity theorem.

\begin{proposition}[Log Potential Continuity Theorem, see \cite{Tao2011}]
\label{prop:log continuity}
If for almost every $z\in\C$, $f_N(z)$ converges almost surely to
\begin{align}
f(z)\ :=\ \int_{\C}\log\abs{w-z}d\mu(w)
\end{align}
for some probability measure $\mu$, then $\mu_{\frac{1}{\sqrt{N}}A_N}$ converges almost surely to $\mu$ \cite{Tao2011}.
\end{proposition}

Thus to show that $\mu_{\frac{1}{\sqrt{N}}A_N}$ converges almost surely to the uniform measure $\mu^\text{{\rm circ}}$ on the unit disk, it suffices to show that the $\log$ potential $f_N(z)$ converges to the corresponding $\log$ potential of $\mu^\text{{\rm circ}}$. For the checkerboard ensembles, we instead show that the re-scaled measure $\mu_{\frac{1}{R\sqrt{N}}A_N}$ converges to $\mu^\text{{\rm circ}}$, where $R=\sqrt{1-1/k}$.

We can reduce the study of $f_N(z)$ to the spectra of Hermitian matrices by rewriting
\begin{align}
f_N(z)&\ =\ \frac{1}{N}\sum_{j=1}^N \log\abs{\frac{\lambda_j(A_N)}{\sqrt{N}}-z} \nonumber
\\
&\ =\ \frac{1}{N}\log\abs{\mathrm{det}\left(\frac{1}{\sqrt{N}}A_N-zI\right)} \nonumber
\\
&\ =\ \frac{1}{2}\int_0^\infty \log x~d\nu_{N,z}(x),
\end{align}
where $d\nu_{N,z}(x)$ is the spectral measure of the Hermitian matrix \be \left(\frac{1}{\sqrt{N}}A_N-zI\right)^*\left(\frac{1}{\sqrt{N}}A_N-zI\right)\ee for $I$ the $N\times N$ identity matrix. This uses the fact that \be\abs{\det A} \ = \ \prod_{j=1}^N \abs{\lambda_j(A)}\ = \ \prod_{j=1}^N \lambda_j(A^*A)^{1/2}.\ee We will analyze the spectral measure of the Hermitian matrices \be\left(\frac{1}{\sqrt{N}}A_N-zI\right)^*\left(\frac{1}{\sqrt{N}}A_N-zI\right)\ee with the method of moments.

We first have to control some convergence issues for our checkerboard ensembles, which arise from singularities of the logarithm at $0$ and $\infty$. The singularity at $\infty$ is considerably easier to control than the one at $0$, as we will see below. Our result is conditional on the following assumption on the least singular values being sufficiently far from $0$, whose role will be made explicit in the lemma that follows.

\begin{assumption}
\label{assumption:singVal bounds}
We say that a measure $\nu_{N,z}$ satisfies this assumption if
\begin{align}
\lim_{T\rightarrow\infty}\sup_{N\geq 1}\int_{\substack{\abs{\log x}\geq T \\ 0<x\leq1}} \log x ~ d\nu_{N,z}\ = \ 0.
\end{align}
\end{assumption}
\textit{Remark.} In the complex asymmetric case, Tao and Vu showed in 2010 that this assumption is satisfied via a polynomial bound on the least singular value, and a count of the other small singular values via the Talgrand concentration inequality, see \cite{BordenaveChafai2012}. This difficulty of controlling the singularity at $0$ has traditionally been the case with the complex asymmetric ensemble: Girko formulated the logarithmic potential approach in 1984 \cite{Girko1984}, but the circular law for the asymmetric ensemble remained unsolved until the behavior of the least singular values was sufficiently controlled by Tao and Vu in 2010 \cite{TaoVu2010}. Our complex symmetric checkerboard ensemble presents difficulties for the control of the small singular values - the symmetric condition (as opposed to entries being iidrv) causes the determinant to be a quadratic function of the rows (as opposed to linear in the iidrv case), and the checkerboard structure adds further complications. See for example \cite{CostelloTaoVu2006} for a discussion of the complications introduced by imposing a symmetric structure on the matrices. Some recent work has been done on extending the polynomial bound on the least singular value to complex symmetric matrices \cite{Nguyen2011} and asymmetric structured ensembles \cite{Cook2016}, but we are not aware of any adequate generalization's of Tao's and Vu's small singular value count to non iidrv matrices. We can, however, control the singular values at $\infty$ for our checkerboard ensemble.

\begin{lemma}
\label{lemma:log singularities}
Fix $z$. For matrices $A_N$ from either the complex asymmetric ensemble or the complex symmetric $k$-checkerboard ensemble, the convergence $\nu_{N,z}\rightarrow\nu_z$ implies the convergence of the corresponding $\log$ potentials
\begin{align}
\int_0^\infty \log x~d\nu_{N,z}(x)\ \rightarrow\ \int_0^\infty\log x~d\nu_z (x)
\end{align}
\end{lemma}
assuming Assumption \ref{assumption:singVal bounds}.
\begin{proof}
The condition we need is uniform integrability. For a Borel function $f:E\rightarrow\R$ and a sequence $\{\eta_N(x)\}_{N\geq1}$ of probability measures on $\R^+$, we say that $f$ is \emph{uniformly integrable} with respect to that sequence of measures if
\begin{align}
\lim_{T\rightarrow\infty}\sup_{N\geq1}\int_{\abs{f}\geq T} \abs{f}~d\eta_N\ =\ 0.
\end{align}
If $f$ satisfies this condition with respect to the sequence $\{\eta_N(x)\}_{N\geq1}$ and is continuous, and the sequence of measures converges weakly $\eta_N\rightarrow\eta$ for some probability measure $\eta$, then
\begin{align}
\lim_{N\rightarrow\infty} \int_{E} f~d\eta_N\ =\ \int_{E} f~d\eta.
\end{align}
For further detail see \cite{BordenaveChafai2012}.

In our case, we will have (for fixed $z$) $E=\R^+$, $\eta_N=\nu_{N,z}$, $\eta=\nu_z$, and $f(x)=\log x$. Since $\log x$ has singularities at $0$ and $\infty$, in order to satisfy uniform integrability we need to control the behavior of the measures $\nu_{N,z}$ at $0$ and $\infty$. To emphasize this we split the integral:
\begin{align}
\int_{\R^+} \log x~d\nu_{N,z}\ =\ \int_0^1 \log x~d\nu_{N,z}+\int_1^\infty \log x~d\nu_{N,z}.
\end{align}

At infinity, we must treat the asymmetric and complex symmetric checkerboard ensembles differently. For the complex asymmetric ensemble, we note that the squared singular values of $\frac{1}{\sqrt{N}}A_N-zI$ are $O(1)$ with probability 1, which trivially suffices for uniform integrability (for sufficiently large $T$, $\nu_{N,z}$ has $0$ mass wherever $x\geq 1$ and $\log x\geq T$). To control the checkerboard ensemble at infinity, we use Weyl's inequalities for singular values to see that the squared singular values of $\frac{1}{R\sqrt{N}}A_N-zI$ have $(N-k)/N$ mass that is $O(1)$, and $k/N$ mass that is $N/(R^2 k^2)+O\left(N^{1/2}\right)$. This also gives uniform integrability:
\begin{align}
\lim_{T\rightarrow\infty}\sup_{N\geq1}\int_{\substack{\abs{\log x}\geq T \\ x\geq1}} \abs{\log x}~d\nu_{N,z}^\text{{\rm check}}\ \leq\ \lim_{T\rightarrow\infty} C\frac{k}{T}\log T\ =\ 0.
\end{align}

The behavior at the origin (the least singular values) is more difficult to control. However, using Assumption \ref{assumption:singVal bounds}, which is known to be satisfied for the complex asymmetric ensemble \cite{TaoVu2010}, we have uniform integrability for both the the complex asymmetric ensemble and our complex symmetric $k$-checkerboard ensemble.\end{proof}

%The behavior of the two ensembles is more similar near the origin. To bound the measure near zero, we need to bound the least singular value of the ensembles $\frac{1}{\sqrt{N}}A_n-zI$. This can be accomplished with a result of Cook that exhibits a polynomial bound for the least singular value of structured matrices, with some moment conditions\cite{Cook2016}. This suffices for uniform integrability: as above, integrating polynomial decay against $\log x$ will go to zero in the limit.
%\end{proof}

Girko showed that, for the complex asymmetric ensemble, the corresponding $\nu_{N,z}^\text{{\rm asym}}$ converges almost surely to an explicit measure $\nu_z$ that satisfies
\begin{align}
\frac{1}{2}\int_0^\infty \log x~d\nu_z(x)\ =\ \int_{\C}\log\abs{w-z}~d\mu^\text{{\rm circ}}(w),
\end{align}
which shows the circular law, i.e., $\mu_{\frac{1}{\sqrt{N}}A_N}\rightarrow\mu^\text{{\rm circ}}$ almost surely, via $\log$ potential continuity and Lemma \ref{lemma:log singularities}.

Thus, to show that the checkerboard ensembles have an eigenvalue bulk that converges to $\mu^\text{{\rm circ}}$, it suffices to show that our measures converge $\nu_{N,z}^\text{{\rm check}}\rightarrow\nu_z$ as well. This is accomplished as follows.

\begin{theorem}
\label{thm:log Moments}
Let $\nu_{N,z}^\text{{\rm asym}}$ and $\nu_{N,z}^\text{{\rm check}}$ be probability measures with
\begin{align}
\frac{1}{2}\int_0^\infty \log x~d\nu_{N,z}^\text{{\rm asym}}(x)&\ =\ \int_{\C}\log\abs{w-z}~d\mu_{\frac{1}{\sqrt{N}}A_N}^\text{{\rm asym}}(w)
\nonumber\\
\text{and}\hspace{2 cm}
\frac{1}{2}\int_0^\infty \log x~d\nu_{N,z}^\text{{\rm check}}(x)&\ =\ \int_{\C}\log\abs{w-z}~d\mu_{\frac{1}{R\sqrt{N}}A_n}^\text{{\rm check}}(w)
\end{align}
obtained as above, where $\mu_{\frac{1}{\sqrt{N}}A_N}^\text{{\rm asym}}$ and $\mu_{\frac{1}{R\sqrt{N}}A_n}^\text{{\rm check}}$ are the spectral measures of matrices normalized as labeled for $R=\sqrt{\frac{k-1}{k}}$, and where the $A_N$ are drawn from complex asymmetric ensemble and the complex symmetric $k$-checkerboard ensemble, respectively. Then as $N\rightarrow\infty$, $\nu_{N,z}^\text{{\rm asym}}$ and $\nu_{N,z}^\text{{\rm check}}$ both converge to a distribution with the same $r$\textsuperscript{{\rm th}} moments $M_z^{(r)}$ given by
\begin{align}
M_z^{(r)}\ =\ \sum_{j=0}^r c_j^{(r)} \abs{z}^{2j}
\end{align}
for some coefficients $c_j^{(r)}$, where $0\leq c_j^{(r)}\leq 4^r C_r$ with $C_r$ the $r$\textsuperscript{{\rm th}} Catalan number.

Since for almost all $z$, $\nu_{N,z}^\text{{\rm asym}}$ is known to converge as $N\rightarrow\infty$ to a measure $\nu_z$ independent of $N$, this allows us to conclude that $\nu_{N,z}^\text{{\rm check}}\rightarrow\nu_z$ as well.
\end{theorem}

\begin{proof}
For any matrix $A_N$ (here $A_N$ can denote either a complex asymmetric matrix, or a hollow complex checkerboard matrix) write $$B_{N,z}\ :=\ \left(\frac{1}{\sqrt{N}}A_N-zI\right)^*\left(\frac{1}{\sqrt{N}}A_N-zI\right).$$ When studying the measures $\nu_{N,z}^\text{{\rm check}}$, it suffices to consider instead ``hollow'' checkerboard matrices $A_N$ where, for $i\equiv j\pmod{k}$ we set $a_{i,j}=0$ instead of $1$. This is because the original checkerboard ensemble is a rank $k$ perturbance of this modified ensemble, i.e., we can write $A_N=A_N^\text{hollow}+P$ for a finite rank matrix $P$, which will also amount to a finite rank perturbance of the matrix $B_{N,z}$ to the analogous $B_{N,z}^\text{hollow}$. Then by \cite[Theorem~1.3]{RMT2016} the two measures converge to the same distribution almost surely.

Returning to the general setting, note that $B_{N,z}$ has entries
\begin{align}
b_{ij}&\ =\ \left(\sum_{m=1}^N \frac{1}{\sqrt{N}}a_{im}-\delta_{im}z\right)^*\left( \frac{1}{\sqrt{N}}a_{mj}-\delta_{mj}z\right) \nonumber
\\
&\ =\ \frac{1}{N}\sum_{m=1}^N \overline{a_{mi}}a_{mj}-\frac{1}{\sqrt{N}}a_{ij}\overline{z}-\frac{1}{\sqrt{N}}\overline{a_{ji}}z+\delta_{ij}|z|^2
\label{eq:expanded entries}
\end{align}
for $\delta_{ij}$ the Kronecker delta ($\delta_ij = 1$ if $i=j$ and is 0 otherwise).

By the eigenvalue trace lemma, the $r$\textsuperscript{{\rm th}} moment $M_{N,z}^{(r)}$ of $\nu_{N,z}$ is given by the following \emph{cyclic product}:
\begin{align}
M_{N,z}^{(r)}\ =\ \frac{1}{N}\sum_{1\leq i_1,\ldots,i_r\leq N} \E[b_{i_1 i_2} b_{i_2 i_3}\cdots b_{i_r i_1}].
\end{align}

By linearity of expectation we can expand the above expectation in terms of $a_{ij}$ using \eqref{eq:expanded entries} and re-indexing. This is done via the following notation. Take $\Psi:\{1,2,\ldots,r\}\rightarrow\{\alpha,\beta,\gamma,\delta\}$ to represent an arbitrary map of sets, and write $A=\card{\Psi^{-1}(\alpha)}$, $B=\card{\Psi^{-1}(\beta)}$, $C=\card{\Psi^{-1}(\gamma)}$, and $D=\card{\Psi^{-1}(\delta)}$. Then $M_{N,z}^{\text{asym~}(r)}$ and $M_{N,z}^{\text{check~}(r)}$ can be expressed by
\begin{align}
\label{hermitianized moments}
M_{N,z}^{\text{asym~}(r)}&\ =\ \sum_{\Psi} \left(\frac{(-1)^{B+C}\overline{z}^{B}z^C\abs{z}^{2D}}{N^{A+B/2+1}}\sum_{\substack{1\leq i_1,\ldots,i_r\leq N \\ 1\leq m_j\leq N \text{~for~} j\in\Psi^{-1}(\alpha)}}\E\left[\xi_1^{(\Psi(1))} \xi_2^{(\Psi(2))} \cdots \xi_r^{(\Psi(r))}\right]\right)
\nonumber\\
M_{N,z}^{\text{check~}(r)}&\ =\ \sum_{\Psi} \left(\frac{(-1)^{B+C}\overline{z}^{B}z^C\abs{z}^{2D}}{R^{2A+B}N^{A+B/2+1}}\sum_{\substack{1\leq i_1,\ldots,i_r\leq N \\ 1\leq m_j\leq N \text{~for~} j\in\Psi^{-1}(\alpha)}}\E\left[\xi_1^{(\Psi(1))} \xi_2^{(\Psi(2))} \cdots \xi_r^{(\Psi(r))}\right]\right)
\end{align}
where
\begin{align}
\xi_j^{(\Psi(j))}\ =\
\begin{cases}
\overline{a_{m_j i_j}}a_{m_j i_{j+1}} & \Psi(j)=\alpha
\\
a_{i_j i_{j+1}} & \Psi(j)=\beta
\\
\overline{a_{i_{j+1} i_j}} & \Psi(j)=\gamma
\\
\delta_{i_j i_{j+1}} & \Psi(j)=\delta.
\end{cases}
\end{align}

We also define notation to write down the conjugacy structure of the above expectation. That is, for fixed $\Psi$, after we factor out the Kronecker deltas $\delta_{i_j i_{j+1}}$ by collapsing the indices $i_j$ and $i_{j+1}$ together, the cyclic product above will become a product of $2A+B$ terms of the form $a_{ij}$ (entries of $A_N$) or $\overline{a_{ij}}$ (conjugates of entries of $A_N$.) Each $a_{ij}$ or $\overline{a_{ij}}$ will be referred to as a \emph{term}. Then, for $1\leq j'\leq 2A+B$ ($j'$ indices meant to enumerate $m_j$ indices as well), let $\psi^+$ denote those indices $j'$ such that the $j'$\textsuperscript{{\rm th}} term in the product is an entry from $A_N$, and let $\psi^-$ denote those indices $j'$ such that the $j'$\textsuperscript{{\rm th}} term in the product is an entry from $\overline{A_N}$.

\begin{example}
\label{ex:Psi function}
The function $\{1,2,3,4,5,6,7\}\xrightarrow{\Psi} \{\beta,\beta,\alpha,\delta,\gamma,\gamma,\alpha\}$ corresponds to the cyclic product $\mathbb{E}\left[\left(a_{i_1 i_2}\right) \left(a_{i_2 i_3}\right) \left(\overline{a_{m_3 i_3}} a_{m_3 i_4}\right) \left(\delta_{i_4 i_5}\right) \left(\overline{a_{i_6 i_5}}\right) \left(\overline{a_{i_7 i_6}}\right) \left(\overline{a_{m_7 i_7}} a_{m_7 i_1}\right) \right]$.
\end{example}

We can now give a different characterization of the above expectations. Below, we use the phrase \emph{step} to refer to each traversal of an edge in a walk on a graph. First we specialize to the complex asymmetric case.

\begin{lemma}
\label{lemma:tree, complex}
For $A_N$ from the complex asymmetric ensemble and for fixed $\Psi$, as $N\rightarrow\infty$ the expectation
\begin{align}
\frac{1}{N^{A+B/2+1}}\sum_{\substack{1\leq i_1,\ldots,i_r\leq N \\ 1\leq m_j\leq N \text{~for~} j\in\Psi^{-1}(\alpha)}}\E\left[\xi_1^{(\Psi(1))} \xi_2^{(\Psi(2))} \cdots \xi_r^{(\Psi(r))}\right].
\end{align}
counts the number of non-isomorphic closed walks on trees with $A+B/2+1$ nodes such that
\begin{itemize}
\item if $B/2$ is not an integer then this quantity is understood to be zero,
\item each edge is traversed exactly twice (once in each direction),
\item each step is given a \emph{sign} $\pm$, and
\item if some edge is traversed first on the $j'_1$\textsuperscript{{\rm th}} step and then later on the $j'_2$\textsuperscript{{\rm th}} step, then either $j'_1$ has positive sign and $j'_2$ has negative sign, or vice versa.
\end{itemize}
We refer to such walks as \emph{signed closed walks on trees}.
\end{lemma}

\begin{proof}
First, note that in such a cyclic product there will be $2A+B$ terms and $2A+B$ indices, once we have collapsed the indices corresponding to Kronecker delta terms $\delta_{ij}$. Suppose we set some of the indices to be equal, so that there are $\ell$ free indices for $1\leq \ell\leq 2A+B$. If $\ell< A+B/2+1$, then the contribution vanishes in the limit $N\rightarrow\infty$. Else, define a graph on $\ell$ vertices by drawing an edge between vertices $i$ and $j$ if there exists a term with indices $i$ and $j$ (respecting the given identification of vertices/indices). There must be at least $\ell-1$ edges, and if $\ell>A+B/2+1$, then there are more than $A+B/2$ distinct edges. By construction, two distinct edges are given by two independent terms. Thus, there are more than $A+B/2$ terms in our cyclic product, so that there exists one term independent from the rest. Note that all terms are drawn from mean $0$ distributions, so that if any one term is independent from the rest, then the entire expectation is immediately zero. So we conclude that, for a nonzero expectation, our graph will contain $A+B/2+1$ vertices and $A+B/2$ edges, and thus is a tree on $A+B/2+1$ nodes. We show first that every cyclic product of nonzero expectation corresponds to a walk satisfying the above conditions.

\begin{figure}[h]
\centering
\scalebox{2}
{
\begin{tikzpicture}
\node (A) at (-1,1){};
\node (B) at (-2, 0){};
\node (C) at (0, 0){};
\node (D) at (-1, -1){};
\node (E) at (1, -1){};
\node (Alb) at (-1,1.2){$\scriptscriptstyle i_1=i_7$};
\node (Blb) at (-2,-.3){$\scriptscriptstyle m_7$};
\node (Clb) at (.7,.2){$\scriptscriptstyle i_2=m_3=i_6$};
\node (Dlb) at (-1,-1.3){$\scriptscriptstyle i_3$};
\node (Elb) at (1,-1.3){$\scriptscriptstyle i_4=i_5$};
\node (e1) at (-.55,.35){\scalebox{.4}{1}};
\node (e2) at (-.65,-.45){\scalebox{.4}{2}};
\node (e3) at (-.45,-.65){\scalebox{.4}{3}};
\node (e4) at (.45,-.65){\scalebox{.4}{4}};
\node (e5) at (.65,-.45){\scalebox{.4}{5}};
\node (e6) at (-.35,.55){\scalebox{.4}{6}};
\node (e7) at (-1.55,.65){\scalebox{.4}{7}};
\node (e8) at (-1.45,.35){\scalebox{.4}{8}};

\filldraw[fill=black] (A) circle (1mm);
\filldraw[fill=black] (B) circle (1mm);
\filldraw[fill=black] (C) circle (1mm);
\filldraw[fill=black] (D) circle (1mm);
\filldraw[fill=black] (E) circle (1mm);
\path[line width=2]
    (-1,1) edge (-2, 0)
    (-1,1) edge (0, 0)
    (0, 0) edge (-1, -1)
    (0, 0) edge (1, -1)
    (-1,1) edge (0,0);

\path[-,line width=.5]
    (-.9,.7) edge (-.6,.4);
\path[-,-{Straight Barb[angle=60:2pt 2]}, line width=.5]
    (-.5,.3) edge (-.2,0);

\path[-,line width=.5]
    (-.3,-.1) edge (-.6,-.4);
\path[-,-{Straight Barb[angle=60:2pt 2]}, line width=.5]
    (-.7,-.5) edge (-1,-.8);

\path[densely dotted,line width=.5]
    (-.8,-1) edge (-.5,-.7);
\path[densely dotted,-{Straight Barb[angle=60:2pt 2]}, line width=.5]
    (-.4,-.6) edge (0,-.2);

\path[-,line width=.5]
    (.1,-.3) edge (.4,-.6);
\path[-,-{Straight Barb[angle=60:2pt 2]}, line width=.5]
    (.5,-.7) edge (.8,-1);

\path[densely dotted,line width=.5]
    (1,-.8) edge (.7,-.5);
\path[densely dotted,-{Straight Barb[angle=60:2pt 2]}, line width=.5]
    (.6,-.4) edge (.2,0);

\path[densely dotted,line width=.5]
    (0,.2) edge (-.3,.5);
\path[densely dotted,-{Straight Barb[angle=60:2pt 2]}, line width=.5]
    (-.4,.6) edge (-.8, 1);

\path[densely dotted,line width=.5]
    (-1.2,1) edge (-1.5,.7);
\path[densely dotted,-{Straight Barb[angle=60:2pt 2]}, line width=.5]
    (-1.6,.6) edge (-2,.2);

\path[-,line width=.5]
    (-1.8,0) edge (-1.5,.3);
\path[-,-{Straight Barb[angle=60:2pt 2]}, line width=.5]
    (-1.4,.4) edge (-1,.8);

%\path[->>, line width=.5]
%    (-.9,.7) edge (-.2,0);
%\path[->, line width=.5]
%    (-.3,-.1) edge (-1,-.8)
%    (-.8,-1) edge (0,-.2)
%    (.1,-.3) edge (.8,-1)
%    (1,-.8) edge (.2,0)
%    (0,.2) edge (-.8, 1)
%    (-1.2,1) edge (-2,.2)
%    (-1.8,0) edge (-1,.8);
%\draw[myarr] (0,0) -- (2,0) node[nodarr] {1};
\end{tikzpicture}
}
\caption{The graph associated to the cyclic product from Example \ref{ex:Psi function} if we consider the matching $i_1=i_7$, $i_2=m_3=i_6$, and $i_4=i_5$ (which is valid, i.e., gives nonzero expectation). The arrows represent steps of the walk, numbered in order, with dotted arrows representing steps of negative sign and solid arrows representing steps of positive sign.}
\label{fig:tree}
\end{figure}
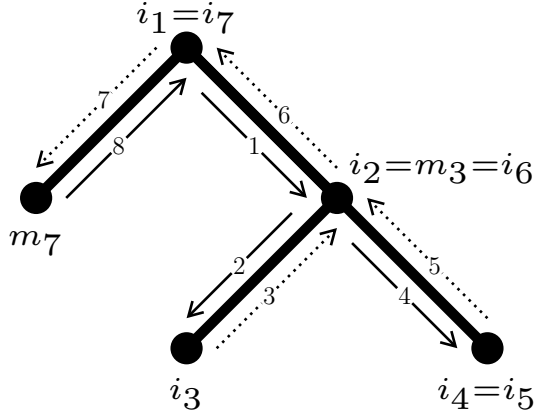

Define a walk on this tree as follows (see Figure \ref{fig:tree} for a useful example). Start at the index corresponding to $i:=i_1$. If the first term is $a_{ij}$, walk from vertex $i$ to vertex $j$, and assign this step a $+$ sign. If instead the first term is $\overline{a_{ji}}$, also walk from vertex $i$ to vertex $j$, but assign this step a $-$ sign. These are the only possibilities. By construction, the cyclic product above will yield a walk on the constructed tree with $A+B/2$ edges. Since we started with $2A+B$ edges, each edge is traversed exactly twice. For the third condition, note by the computations in \eqref{eq:complex expectation1} and \eqref{eq:complex expectation2} that two dependent terms (i.e., the two times we traverse the same edge) with a nonzero expected product must be of different conjugacy classes (i.e., one is an entry from $A_N$, and another is an entry from $\overline{A_N}$). This translates to the third condition above: since the sign of each step keeps track of the conjugacy class that gave that step, this means that the two steps for every edge must have opposite sign.

Now we claim that every walk on a tree with $A+B/2$ nodes satisfying the above conditions corresponds to a cyclic product with nonzero expectation. Given such a walk, with first step from $i$ to $j$, set the first term in the cyclic product to be $a_{ij}$ if the step has sign $+$, and $\overline{a_{ji}}$ if the step has sign $-$. Doing this for all steps will generate a cyclic product where every term is paired with exactly one other term, in a conjugate pair. To ensure we count only terms of nonzero expectation, we need to check that all conjugate pairs are of the form $\E[a_{ij}\overline{a_{ij}}]=1$ and not $\E[a_{ij}\overline{a_{ji}}]=0$, since we are working with a complex asymmetric ensemble, and $a_{ij}$ is independent from $a_{ji}$. However, the second situation cannot happen, since our walk occurs on a tree, i.e., if each edge is traversed exactly twice, the two steps must occur in opposite directions, so one step goes from $i$ to $j$, while the other goes from $j$ to $i$. Then, forcing the steps to have opposite signs implies our paired expectations are of the form $\E[a_{ij}\overline{a_{ij}}]=1$. This shows the bijection between our cyclic expectations and walks satisfying the above conditions.

Then, once we have fixed a valid identification of the starting $2A+B$ indices to $A+B/2+1$ free indices, each free index has $N$ choices (un-identified indices being assigned the same index will happen non-generically as $N\rightarrow\infty$ for $r$ fixed, and will vanish in the limit as lower order terms by degree of freedom counts as above), so each nonzero expectation will contribute $N^{A+B/2+1}$ which equals $1$ after dividing out by the normalizing factors of $N$ present above, which suffices for the lemma.
\end{proof}

\begin{corollary}
\label{cor:checkerboard tree}
For $A_N$ from the complex symmetric $k$-checkerboard ensemble and for fixed $\Psi$, Lemma \ref{lemma:tree, complex} holds for the analogous expectation:

\begin{align}
\frac{1}{R^{2A+B}N^{A+B/2+1}}\sum_{\substack{1\leq i_1,\ldots,i_r\leq N \\ 1\leq m_j\leq N \text{~for~} j\in\Psi^{-1}(\alpha)}}\E\left[\xi_1^{(\Psi(1))} \xi_2^{(\Psi(2))} \cdots \xi_r^{(\Psi(r))}\right].
\end{align}
\end{corollary}

\begin{proof}
The interpretation of cyclic products as walks as above for $a_{ij}$ entries of checkerboard matrices is similarly valid up to the last two paragraphs of the proof of Lemma \ref{lemma:tree, complex}.

It is clear that every identification of indices that gives a nonzero expectation for the asymmetric ensemble also gives a nonzero expectation for the symmetric ensemble. We should check that we can find no additional matchings from the symmetric condition. This is something we have already seen: above, we argued given a walk satisfying the above conditions and the associated cyclic product it generates, \emph{none} of the pairings in the cyclic product are of the form $\E[a_{ij}\overline{a_{ji}}]$, because of the restriction that steps on the same edge are of opposite sign.

However, we have an additional restriction that when choosing our indices, if there exists an edge between $i$ and $j$, then $i\not\equiv j\pmod{k}$. This is because if $i\equiv j\pmod{k}$ then $a_{ij}=0$, and the expectation is zero. The modification is clear from the interpretation of walks on a tree: since trees are acyclic, once we have fixed the congruence class of one vertex, (recalling $\ell=A+B/2+1$ there will be $(k-1)^{\ell-1}$ ways to choose congruence classes of the other vertices such that adjacent vertices on the tree do not share the same congruence class, so there are $k(k-1)^{\ell-1}$ ways to choose the congruence classes. After fixing the congruence classes there are $(N/k)^\ell$ ways to choose indices for each vertex, so we see there are $N^\ell\left(\frac{k-1}{k}\right)^{\ell-1}=N^{A+B/2+1}R^{2A+B}$ ways to choose the vertices that will give a nonzero expectation. Dividing by the normalizing factors of $N$ and $R$ in front then give the result.
\end{proof}

This shows that for all $z\in\C$, $\lim_{N\rightarrow\infty} M_{N,z}^{\text{asym~}(r)}=\lim_{N\rightarrow\infty} M_{N,z}^{\text{check~}(r)}$. A few observations reduce the moments to the form claimed in Theorem \ref{thm:log Moments}. Since above we saw that every term must be paired with exactly one other conjugate term to produce a nonzero expectation, we conclude $B=C$, so our moments must take the form claimed in Theorem \ref{thm:log Moments}.

\begin{remark} Note that the proof of Corollary \ref{cor:checkerboard tree} also goes through for the complex asymmetric $k$-checkerboard ensemble. Because of the restriction of matching terms in conjugate pairs, the symmetry condition is immaterial for this specific bulk calculation as is discussed above.
\end{remark}

\begin{corollary}
\label{cor:moment bound}
We have $0\leq c_j^{(r)}\leq 4^r C_r$ with $C_r$ the $r$\textsuperscript{{\rm th}} Catalan number.
\end{corollary}
\begin{proof}
Note that there are $4^r$ choices of $\Psi$. Then, since $A+B/2+1\leq r+1$ and the number of non-isomorphic closed walks on trees with $r+1$ nodes such that each edge is traversed exactly twice is given by $C_r$ (equivalently the number of ordered trees on $r+1$ nodes), we obtain the claimed upper bound via Lemma \ref{lemma:tree, complex} and Equation \eqref{hermitianized moments}. Note that this upper bound is far from tight, as we have not removed any walks corresponding to zero expectation, from the assignment of signs to steps.
\end{proof}

\begin{example}
The coefficients $c_j^{(r)}$ can be computed by hand. For example, there is one choice of $\Psi$ such that $B+C+2D=0$ (i.e., $A=r$), which corresponds to signed closed walks on trees with $r$ nodes. Then, as in the proof of Theorem \ref{thm:singVal bulk}, every closed walk on a tree where every even step has sign $+$ and every odd step has sign $-$ always satisfies the condition that the two steps on the same edge are of opposite sign, i.e., the number of valid walks is counted by the Catalan number $C_r$, so we conclude $c_0^{(r)}=C_r$. There is also only one choice of $\Psi$ such that $B+C+2D=2r$ (i.e., $D=r$), which corresponds to signed closed walks on trees with one node, of which there is one, so $c_r^{(r)}=1$.
\end{example}

\begin{example}
The first few moments are
\begin{align}
M_z^{(1)}&\ =\ 1+\abs{z}^2
\nonumber\\
M_z^{(2)}&\ =\ 2+3\abs{z}^2+\abs{z}^4
\nonumber\\
M_z^{(3)}&\ =\ 5+15\abs{z}^2+6\abs{z}^4+\abs{z}^6.
\end{align}
\end{example}

In particular, Corollary \ref{cor:moment bound} implies that Carleman's condition is satisfied for our moments at each fixed $z$, since for fixed $z$ the moments are bounded by some scaling of the Catalan numbers (which satisfy Carleman's condition as in Wigner's semicircle law). Thus the moments uniquely characterize the distribution for every fixed $z$, i.e., $\nu_{N,z}^\text{{\rm check}}\rightarrow\nu_z$ and we have proved Theorem \ref{thm:log Moments}.
\end{proof}

\begin{proof}[Proof of Theorem \ref{thm:complexSym bulk}]
Applying Proposition \ref{prop:log continuity} and Lemma \ref{lemma:log singularities} to Theorem \ref{thm:log Moments} suffices for the proof of Theorem \ref{thm:complexSym bulk}.
\end{proof}

%--------------------------------%

%%%%%%%%%%%%%%%%%%%%%%%%%%%%%%%%%%%%%%%%%%%%%%%%%%%%%%%%
%%%%%%%%%%%%%%%%%%%%%%%%%%%%%%%%%%%%%%%%%%%%%%%%%%%%%%%%
%%%%%%%%%%%%%%%%%%%%%%%%%%%%%%%%%%%%%%%%%%%%%%%%%%%%%%%%
%%%%%%%%%%%%%%%%%%%%%%%%%%%%%%%%%%%%%%%%%%%%%%%%%%%%%%%%

%--------------------------------%
%%%%%%%%%%%%%%%%%%%%%%%%%%%%%%%%%%%%%%%%%%%%%%%%%%%%%%%%%%%%%%%%%%%%%%%%%%%%%%%%%%%%%%%%%%%%%%%%%%%%%%%%%%%%%%%%
%%%%%%%%%%%%%%%%%%%%%%%%%%%%%%%%%%%%%%%%%%%%%%%%%%%%%%%%%%%%%%%%%%%%%%%%%%%%%%%%%%%%%%%%%%%%%%%%%%%%%%%%%%%%%%%%
%%%%%%%%%%%%%%%%%%%%%%%%%%%%%%%%%%%%%%%%%%%%%%%%%%%%%%%%%%%%%%%%%%%%%%%%%%%%%%%%%%%%%%%%%%%%%%%%%%%%%%%%%%%%%%%%
%%%%%%%%%%%%%%%%%%%%%%%%%%%%%%%%%%%%%%%%%%%%%%%%%%%%%%%%%%%%%%%%%%%%%%%%%%%%%%%%%%%%%%%%%%%%%%%%%%%%%%%%%%%%%%%%
\section{Generalized checkerboard ensembles}
\label{section:genCheck}

%%%%%%%%%%%%%%%%%%%%%%%%%%%%%%%%%%%%%%%%%%%%%%%%%%%%%%%%%%%%%%%%%%%%%%%%%%%%%%%%%%%%%%%%%%%%%%%%%%%%%%%%%%%%%%
%%%%%%%%%%%%%%%%%%%%%%%%%%%%%%%%%%%%%%%%%%%%%%%%%%%%%%%%%%%%%%%%%%%%%%%%%%%%%%%%%%%%%%%%%%%%%%%%%%%%%%%%%%%%%%
%%%%%%%%%%%%%%%%%%%%%%%%%%%%%%%%%%%%%%%%%%%%%%%%%%%%%%%%%%%%%%%%%%%%%%%%%%%%%%%%%%%%%%%%%%%%%%%%%%%%%%%%%%%%%%
\subsection{Analogs of bulk and blip results for generalized checkerboard ensembles}\ \\

The ensembles that follow in this subsection will have random variables iidrv (up to a symmetry restriction).  Many of the main results above hold for generalized checkerboard ensembles as well.

%, and deterministic entries all $1$ for simplicity.

\begin{proof}[Proof of Corollary \ref{cor:gen singVal bulk}]
As in the proof of Theorem \ref{thm:singVal bulk} we can reduce to the case where the deterministic entries are all zero by a perturbation argument, apply the eigenvalue trace lemma up to Equation \eqref{eq:singVal moments}, and also interpret the cyclic products as closed walks on trees with $r+1$ nodes. The number of closed walks traversing each edge twice on trees with $r+1$ nodes is again given by the number of ordered trees on $r+1$ nodes, which is $C_r$. As before, it is not true that all $N^{r+1}+O\left(N^r\right)$ choices of indices will yield a nonzero expectation in the cyclic product, since if $i$ and $j$ are adjacent indices then we cannot have $a_{ij}=0$. Before, this reduced to the condition $i\not\equiv j\pmod{k}$. In the $m$-regular $k$-checkerboard case, however, the correct condition is that, having fixed any congruence class modulo $k$ of some index $i$, there are $k-m$ congruence classes for the index $j$ such that $a_{ij}\not=0$. This is because $m$ is a constant dependent only on the entire ensemble. Thus we see that $(R/2)^{2r}N^{r+1}+O\left(N^r\right)$ choices of indices will yield a nonzero expectation in the cyclic product. Indeed, there are $N$ choices for the root of the tree, $N(1-m/k)$ choices for each of their children, $N(1-m/k)$ choices for each of the children in the next level down, etc., which shows that all moments of the squared singular value bulk are $M_r=\left(\frac{R}{2}\right)^{2r} C_r$ as claimed.
\end{proof}

\begin{example}
In general, a non-regular generalized $k$-checkerboard ensemble need not have singular values following a quarter-circular bulk. Consider the generalized $2$-checkerboard ensemble tiled with $2\times 2$ matrices of the form
\begin{align}
   \begin{pmatrix}
    1 & * \\
    * & * \\
   \end{pmatrix}
\end{align}
where, as before, each $*$ represents an iidrv complex random variable with mean $0$ variance $1/2$ that respects the symmetric structure . As in the proof of Theorem \ref{thm:singVal bulk} the bulk of this ensemble will converge to the bulk of the ensemble with the entries $1$ replaced with $0$. Brute force computing the small moments (counting by hand the appropriate walks on trees as is done to compute the quarter-circular bulk above) gives $M_1=3/4$, $M_2=10/8$, $M_3=42/16$, and $M_4=198/32$. In particular the bulk cannot be quarter-circular of radius $R$, which would correspond to moments $M_r=(R/2)^{2r} C_r$, e.g. $M_1=(1/2)(R^2/2)$, $M_2=(2/4)(R^4/4)$, $M_3=5/8(R^6/8)$, and $M_4=(14/16)(R^8/16)$.\footnote{We note that the numerators (starting with the $0-th$ moment) $2,3,10,42,198$ are the first five terms of the OEIS sequence A007226, which relates to counting certain ternary trees.}
\end{example}

\begin{proof}[Proof of Corollary \ref{cor:gen eigen bulk}]
The discussion of the log potential in the proof of Theorem \ref{thm:complexSym bulk} is done in the same way for our generalized checkerboard ensemble up to Lemma \ref{lemma:log singularities}. The proof of that lemma goes through as well, assuming the appropriate analogy to Assumption \ref{assumption:singVal bounds} (replace instances of the checkerboard ensemble with the generalized checkerboard ensemble), when we note that Weyl's inequality for singular values of $\frac{1}{\sqrt{N}}A_N-zI$ for fixed $z$ gives at least $N-k$ squared singular values of size $O\left(1\right)$, and at most $k$ squared singular values of size $O\left(N\right)$, in which case the argument in Lemma \ref{lemma:log singularities} for the behavior at $\infty$ can be applied as well. The qualifiers ``at least'' and ``at most'' come from the fact that the perturbation from the hollow ensemble (deterministic $1$'s replaced with $0$'s) for a generalized $k$-checkerboard is rank at most $k$, possibly less.

Then, the expansions of Theorem \ref{thm:log Moments} in the checkerboard context can also be done for the generalized checkerboard ensembles, with the necessary modifications arising in an analogous proof of Corollary \ref{cor:checkerboard tree}. As discussed in that proof, the restriction of symmetry on the ensemble does not affect the combinatorics that arise when counting signed closed walks on trees. Again, the modification comes when counting how many ways there are to choose indices, given a walk, since adjacent indices $i,j$ must have $a_{ij}\not=0$ for the expectation to be nonzero as discussed in the proof of Corollary \ref{cor:checkerboard tree}. Here, there are $N$ ways to choose the index of the root, $R^2N$ ways to choose the indices of the root's children, $R^N$ ways to choose the indices of those children's children, etc. which shows as in Corollary \ref{cor:checkerboard tree} that the bulk converges to a uniform disc centered at the origin scaled by $R$.
\end{proof}

%%%%%%%%%%%%%%%%%%%%%%%%%%%%%%%%%%%%%%%%%%%%%%%%%%%%%%%%%%%%%%%%%%%%%%%%%%%%%%%%%%%%%%%%%%%%%%%%%%%%%%%%%%%%%%
%%%%%%%%%%%%%%%%%%%%%%%%%%%%%%%%%%%%%%%%%%%%%%%%%%%%%%%%%%%%%%%%%%%%%%%%%%%%%%%%%%%%%%%%%%%%%%%%%%%%%%%%%%%%%%
%%%%%%%%%%%%%%%%%%%%%%%%%%%%%%%%%%%%%%%%%%%%%%%%%%%%%%%%%%%%%%%%%%%%%%%%%%%%%%%%%%%%%%%%%%%%%%%%%%%%%%%%%%%%%%
\subsection{Eigenvalues of complex checkerboard matrices: blip}\ \\

The main steps in the proof of Theorem \ref{thm:complex blip} are to first restrict our attention to regions $\Omega_\epsilon$ to avoid the singularity at $0$, show that the distribution must be discrete and finitely supported, show that discrete distributions are characterized by (holomorphic) moments, compute moments via the eigenvalue trace lemma, and to control the error arising from restricting to $\Omega_\epsilon$.

\begin{lemma}
\label{lemma:discrete moment}
For all $\epsilon>0$, assume almost sure convergence to some measure $\tilde{\mu}_N\rightarrow\tilde{\mu}$ as measures on $\Omega_\epsilon$. With any fixed $\epsilon>0$, the measure $\tilde{\mu}$ on $\Omega_\epsilon$ must be a discrete measure with finite support.
\end{lemma}

\begin{proof}
Consider $0<\epsilon'<\epsilon$ and write $\tilde{\mu}_N$ and $\tilde{\mu}_N'$ for the restrictions to $\Omega_\epsilon$ and $\Omega_{\epsilon'}$ respectively. Write also $\tilde{\mu}_N\rightarrow\tilde{\mu}$ and $\tilde{\mu}_N'\rightarrow\tilde{\mu}'$ via our convergence assumption. Note $\tilde{\mu}'$ restricts to $\tilde{\mu}$ on $\Omega_\epsilon$.

%Note that any generalized $k$-checkerboard matrix $A$ can be decomposed $A=M+P$ where $M$ is a generalized $k$-checkerboard matrix with all deterministic entries set to $0$, and $P$ is completely deterministic and composed of repeating blocks of some fixed $k\times k$ matrix $B$ (determined by the ensemble).
Write $p(x)$ for the (degree $k$) characteristic polynomial of $B$, and define $f(z):=z p(z)$, viewed as a function $f\colon\C\rightarrow\C$. Note $f$ has fixed degree as $N\rightarrow\infty$.

Consider the pushforward measures $f_*\tilde{\mu}_N'\rightarrow f_*\tilde{\mu}'$ on $f(\Omega_{\epsilon'})$. Note that $f_*\tilde{\mu}_N$ is also the spectral distribution of $f\left(\frac{k}{N} A_N\right)$ restricted to $f(\Omega_{\epsilon'})$ with total (restricted) measure scaled by $N/k$, and in particular has the same support. We can control the support of the spectral distribution of $f\left(\frac{k}{N} A_N\right)$ via its largest singular value, which majorizes all its eigenvalues. If $c$ is the constant term of $p(z)$, the Cayley-Hamilton theorem shows that $p\left(\frac{k}{N}P\right)-c$ is a block matrix consisting of repeating $k\times k$ identity blocks scaled by $-kc/N$. This then shows $f\left(\frac{k}{N} P\right)=0$. Note $f\left(\frac{k}{N} A_N\right)=f\left(\frac{k}{N} P\right)+O\left(\frac{k}{N} M\right)=O\left(\frac{k}{N} M\right)$, where the big $O$ term is a sum of mixed products of $\frac{k}{N}M$ and $\frac{k}{N}P$ at least linear in $\frac{k}{N}M$. The number of terms in this sum is fixed as $N\rightarrow\infty$ because $f$ has fixed degree. The largest singular value of $\frac{k}{N}P$ is $O(1)$ and the largest singular value of $\frac{k}{N} M$ is almost surely $O(N^{-1/2+\epsilon})$ for all $\epsilon>0$ via a standard method of moments argument, see the proof of Proposition \ref{prop:singVal regimes} and Lemma B.3 of \cite{RMT2016} (recall that $M$ has all deterministic entries set to zero).\footnote{This bound alternately follows from the Gershgorin Circle Theorem applied to $M$ -- central limit theorem on the at most $N$ random variables in each row gives Gershgorin discs of radii $O(N^{1/2})$ with centers at distance $O(1)$ from the origin.}

Thus the largest singular value of $f\left(\frac{k}{N}A_N\right)$ is almost surely $O\left(N^{-1/2+\epsilon}\right)$ for all $\epsilon>0$. In particular, as $N\rightarrow\infty$ the spectral measure of $f\left(\frac{k}{N} A_N\right)$ and thus $f_*\tilde{\mu}_N'$ almost surely has support contained in any fixed radius ball at the origin. Considering the neighborhood $f(\mathcal{B}_{\epsilon'})$ of the origin, we find that $f_*\tilde{\mu}_N'$ is almost surely supported on $f(\mathcal{B}_{\epsilon'})$ as $N\rightarrow\infty$. Fixing any $\delta>0$, we can then select $\epsilon'>0$ small such that $f_*\tilde{\mu}_N'$ vanishing outside $f(\mathcal{B}_{\epsilon'})$ implies that, on $\Omega_{\epsilon'}$, $\tilde{\mu}'_N$ vanishes outside balls of radius $\delta$ centered at the zeroes of $f$.\footnote{Note $\C\setminus f(\Omega_{\epsilon'})\subset f(\mathcal{B}_{\epsilon'})$ by surjectivity of $f$.}

Since $\tilde{\mu}_N'$ restricts to $\tilde{\mu}_N$ on $\Omega_{\epsilon}$ as long as $\epsilon'<\epsilon$, sending $\delta\rightarrow 0$ and then $N\rightarrow\infty$ shows that $\tilde{\mu}$, restricted to $\Omega_{\epsilon}$ is finitely supported at the (nonzero) zeroes of $f$.
\end{proof}

We would like to analyze the moments of $\tilde{\mu}_N\rightarrow\tilde{\mu}$. However, the formula from the eigenvalue trace lemma applies to moments of $\tilde{\mu}_N$ on all of $\C$ rather than restricted $\Omega_\epsilon$. We first control this error.

\begin{lemma}
\label{lemma:moment epsilon}
As $N\rightarrow\infty$, the measure $\tilde{\mu}_N$ restricted to $\mathcal{B}_{\epsilon}$ contributes at most $(1/k)(2k+2)\epsilon$ to the $r$\textsuperscript{{\rm th}} moments, for $r\geq 6$.
\end{lemma}

\begin{proof}
Applying Weyl's inequalities as in Proposition \ref{prop:singVal regimes} shows that, almost surely, the $k$ largest singular values of $A_N$ are $O(N)$ and the remaining $N-k$ singular values are $O(N^{1/2+\epsilon})$. Recall that the product of the $m$ largest singular values majorizes the product of $m$ largest eigenvalues. The $m$ product of the $m$ largest eigenvalues of $A_N$ thus have growth $o\left(N^{k+(m-k)3/5}\right)$ as $N\rightarrow\infty$. Thus the total measure of $\tilde{\mu}_N$ outside of $\mathcal{B}_{N^{-1/5}}$ is bounded above by $(1/k)(2k+1)$ as $N\rightarrow\infty$ since $m(4/5)>k+(m-k)(3/5)$ when $m\geq 2k+1$. The measure $\tilde{\mu}_N$ restricted to the region outside of $\mathcal{B}_{N^{-1/5}}$ but within $\mathcal{B}_{\epsilon}$ thus contributes at most $(2k+1)\epsilon$ to any positive moments, while the measure $\tilde{\mu}_N$ restricted to $\mathcal{B}_{N^{-1/5}}$ contributes at most $N\cdot N^{-6/5}=O(N^{-1/5})$ to $r$\textsuperscript{{\rm th}} moments for $r\geq6$ as $N\rightarrow\infty$, which we bound by $(1/k)\epsilon$. Adding these two contribution gives the claimed bound.
\end{proof}

\begin{remark}
\label{remark:finite measure}
The part of the proof of Lemma \ref{lemma:moment epsilon} bounding the total measure outside $\mathcal{B}_{N^{-1/5}}$ also shows that $\tilde{\mu}$ viewed as a measure on $\Omega_\epsilon$ has finite total measure, for all $\epsilon>0$.
\end{remark}

We next show how the (holomorphic) moments characterize discrete distributions.

\begin{lemma}
\label{lemma:discreteMoments}
Any finite discrete measure $\mu$ on $\C$ with finite support in $\C\setminus\{0\}$ is uniquely determined by its $r$\textsuperscript{{\rm th}} integer moments for $r\geq \alpha$, for any fixed integer $\alpha>0$.
\end{lemma}
\begin{proof}
Given distinct nonzero complex numbers $\{z_j\}_{j=1}^n$ with nonzero coefficients $\{\lambda_j\}_{j=1}^n$, and distinct nonzero complex numbers and $\{z_k'\}_{k=1}^m$ with nonzero coefficients $\{\lambda_k'\}_{k=1}^m$ such that
\begin{align}
\label{eq:momMatch}
\sum_{j=1}^n \lambda_j z_j^r\ =\ \sum_{k=1}^m \lambda_k' z_k'^r
\end{align}
for all $r\geq N$, we wish to show that $n=m$ and, for some appropriate permutation of the indices, $\lambda_j=\lambda_k'$ for $j=k$, and $z_j=z_k'$ for $j=k$.

We can associate to any complex number $z$ the sequence $\mathbf{\tilde{z}}=(z^\alpha,z^{\alpha+1},z^{\alpha+2},\ldots)\in\C^{\N}$. Then Equation (\ref{eq:momMatch}) becomes
\begin{align}
\label{eq:seqMatch}
\sum_{j=1}^n \lambda_j \mathbf{z_j}\ =\ \sum_{k=1}^m \lambda_k' \mathbf{z_k'}.
\end{align}
Under the left-shift linear operator in the sequence space $\C^{\N}$, note that $\mathbf{\tilde{z}}$ is a nonzero eigenvector with eigenvalue $z$. In particular, sequences associated with distinct complex numbers have distinct eigenvalues and are thus linearly independent. Applying this to Equation (\ref{eq:seqMatch}) gives the claim.
\end{proof}

\begin{lemma}
The expected $r$\textsuperscript{{\rm th}} moments $\E\left[\tilde{M}_N^{(r)}\right]$ converge to the $r$\textsuperscript{{\rm th}} moment of the spectral measure of the deterministic $k\times k$ matrix $B$ as $N\rightarrow\infty$.
\end{lemma}
\begin{proof}
Write $\tilde{M}_N^{(r)}$ for the $r$\textsuperscript{{\rm th}} moment of $\tilde{\mu}_N$. By the eigenvalue trace lemma, the expected $r$\textsuperscript{{\rm th}} moment is given by
\begin{align}
\E\left[\tilde{M}_N^{(r)}\right]\ =\ \frac{N}{k}\frac{1}{N}\frac{k^r}{N^r}\sum_{1\leq i_1,\ldots,i_r\leq N} \E\left[a_{i_1 i_2}\cdots a_{i_r i_1}\right].
\end{align}
We analyze this with a standard degree of freedom count. Fix a cyclic product $\E\left[a_{i_1 i_2} \cdots a_{i_r i_1}\right]$. Consider some $a_{ij}$ appearing in this cyclic product that corresponds to a random variable, i.e. not a deterministic entry of the matrix. If $a_{ij}$ is independent from the the other entries, the contributed expectation is zero since the random variables in the matrix ensemble have mean $0$. Else, $a_{ij}$ matches another term in the cyclic product, which loses at least one degree of freedom for the indices. In particular, such a contribution to $\E\left[\tilde{M}_N^{(r)}\right]$ must be $O(1/N)$ which vanishes as $N\rightarrow\infty$.

Thus, as $N\rightarrow\infty$, the quantity $\E\left[\tilde{M}_N^{(r)}\right]$ can be computed only considering cyclic products with all entries deterministic. Reducing modulo $k$ (and recalling that the deterministic entries of $k$-checkerboard matrices repeat modulo $k$)
\begin{align}
\E\left[\tilde{M}_N^{(r)}\right]\ =\ \frac{1}{k}\sum_{1\leq i_1,\ldots,i_r\leq k} \E\left[a_{i_1 i_2}\cdots a_{i_r i_1}\right]
\end{align}
where the expectations $\E\left[a_{i_1 i_2}\cdots a_{i_r i_1}\right]$ are understood to be taken over deterministic entries only. This is precisely the $r$\textsuperscript{{\rm th}} moment of the spectral measure of the matrix $B$ as claimed.
\end{proof}

\begin{proof}[Proof of Theorem \ref{thm:complex blip}]
Fix $\epsilon>0$ and consider $0<\epsilon'<\epsilon$. By Lemma \ref{lemma:moment epsilon} the $r$\textsuperscript{{\rm th}} moments of $\tilde{\mu}$ on $\Omega_{\epsilon'}$ differ from the $r$\textsuperscript{{\rm th}} moments of the spectral measure of $B$ by at most $(1/k)(2k+2)\epsilon'$, when $r\geq6$. Since $\tilde{\mu}$ on $\Omega_\epsilon$ restricts to $\tilde{\mu}$ on $\Omega_{\epsilon'}$ when $\epsilon'<\epsilon$, and $\tilde{\mu}$ is supported on the zeroes of $f$, the $r$\textsuperscript{{\rm th}} moments of $\tilde{\mu}$ on $\Omega_\epsilon$ are equal to the $r$\textsuperscript{{\rm th}} moments of $\tilde{\mu}$ on $\Omega_{\epsilon'}$ whenever $\epsilon$ is smaller than the smallest nonzero eigenvalue of $B$ (i.e. smallest nonzero zero of $f$). Sending $\epsilon'\rightarrow 0$ shows that the $r$\textsuperscript{{\rm th}} moments of $\tilde{\mu}$ on $\Omega_\epsilon$ must be equal to the $r$\textsuperscript{{\rm th}} moments of the spectral measure of $B$ restricted to $\Omega_\epsilon$ when $r\geq 6$. By Lemma \ref{lemma:discrete moment}, $\tilde{\mu}$ on $\Omega_\epsilon$ is a discrete measure with finite support at the nonzero zeroes of $f$. By Remark \ref{remark:finite measure} the total measure is finite. Lemma \ref{lemma:discreteMoments} shows that these two measures must be equal as claimed.
\end{proof}

\begin{proof}[Proof of Corollary \ref{cor:bulk mass}]
By Theorem \ref{thm:complex blip}, the $k'$ largest eigenvalues of $A_N$ are all of size $O(N)$. Applying Weyl's inequalities as in Proposition \ref{prop:singVal regimes} shows that, almost surely, the $k'$ largest singular values of $A_N$ are $O(N)$ and the remaining $N-k'$ singular values are $O(N^{1/2+\delta})$. The product of the $m$ largest singular values majorize the product of the $m$ largest eigenvalues, which suffices for the claim with a similar argument as in the proof of Lemma \ref{lemma:moment epsilon}.
\end{proof}

%%%%%%%%%%%%%%%%%%%%%%%%%%%%%%%%%%%%%%%%%%%%%%%%%%%%%%%%%%%%%%%%%%%%%%%%%%%%%%%%%%%%%%%%%%%%%%%%%%%%%%%%%%%%%%
%%%%%%%%%%%%%%%%%%%%%%%%%%%%%%%%%%%%%%%%%%%%%%%%%%%%%%%%%%%%%%%%%%%%%%%%%%%%%%%%%%%%%%%%%%%%%%%%%%%%%%%%%%%%%%
%%%%%%%%%%%%%%%%%%%%%%%%%%%%%%%%%%%%%%%%%%%%%%%%%%%%%%%%%%%%%%%%%%%%%%%%%%%%%%%%%%%%%%%%%%%%%%%%%%%%%%%%%%%%%%
\subsection{Conjectures}\ \\
\label{section:conjectures}

%\begin{figure}[h]
%\scalebox{.3}{\includegraphics{Smiley.eps}}
%\caption{Numerical example of a conjectured possible distribution using generalized $k$-checkerboard matrices.}
%\label{fig:smiley}
%\end{figure}

Although in Section \ref{section:genCheck} we only analyzed ensembles resulting in ``discrete-type'' blip distributions, it is natural to ask whether we can naturally construct other ensembles where the resulting blip distribution will not be discrete. For example, if we modify the complex $k$-checkerboard ensemble and replace the deterministic $1$'s with complex numbers on the unit circle, drawn with uniform probability, our result trivially implies that the blip will consist of a ring of eigenvalues on the unit circle, in the same sense of Theorem \ref{thm:complex blip}; see Figure \ref{fig:imagestoseelater} (right).

%\begin{figure}[h]
%\scalebox{.5}{\includegraphics{SaturnV.eps}}
%\caption{A ring of eigenvalues.}
%\end{figure}

One can attempt to construct non-discrete blip distributions in other ways, for example with an analogy of a generalized complex $k$-checkerboard ensemble with $B$ matrix having eigenvalues at the $k$-th roots of unity, except where $k=\sqrt{N}$ is growing as $N\rightarrow\infty$ over the squares. Heuristically, one expects the blip distribution to be a ring of vanishing thickness, in some sense, but one would need different techniques than those used to describe the discrete-type blip distributions characterized by Theorem \ref{thm:complex blip}.

%The observed behavior of $k$-checkerboard matrices seems to be much more patterned than current techniques allow us to rigorously prove. In both the real and complex cases, we are able to control the distribution of blips numerically. Namely, the first $k\times k$ submatrix of a $k$-checkerboard matrix seems to determine a blip distribution. The distribution of eigenvalues is simply the distribution of $\frac{N}{k}$ times the eigenvalues of $k\times k$, distributed based on the random entries of the $k\times k$ matrix. Of course, this term is not always visible, depending on what generalized $k$-checkerboard pattern is chosen. This conjecture has interesting implications for what distributions are possible; one such distribution is displayed in Figure \ref{fig:smiley}. By allowing $k$ to be a slow growing function of $n$, one may be able to construct arbitrary blip distributions, although care is required in rigorizing this notion.

%\begin{figure}[h]
%\scalebox{.5}{\includegraphics{Saturn1.eps}}
%\caption{Disaster 1, 300x300 matrices}
%\end{figure}

%\begin{figure}[h]
%\scalebox{.5}{\includegraphics{Saturn2.eps}}
%\caption{Disaster 2, I think with a different $\gamma$. It's almost %like part of the ring is accreting into the bulk...}
%\end{figure}

As a final note, the bulk also seems to deviate from standard circular law behavior in more general ensembles when, for example the entries are no longer iidrv, and different entries are assigned different means or variances. Instead we observe some sharpened bulk distributions, that are distinctly non-uniform.

\begin{figure}[h]
    \centering
    \begin{minipage}{0.45\textwidth}
        \centering
        \scalebox{.6}{\includegraphics{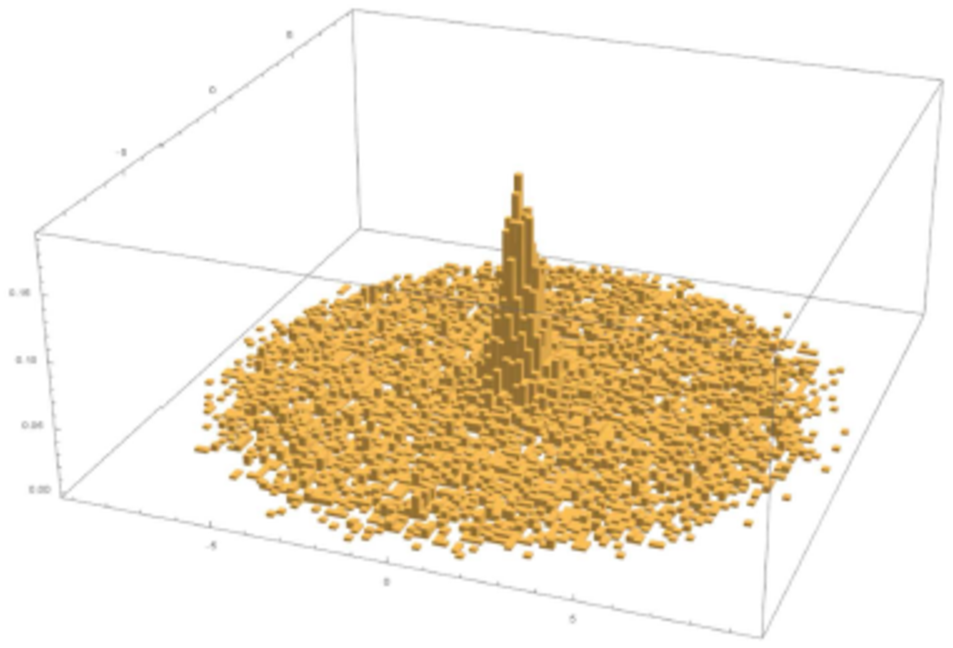}}
        \caption{Sharp bulk distribution, two different variance values for the random entries.}
    \end{minipage}\hfill
    \begin{minipage}{0.45\textwidth}
        \centering
\scalebox{.6}{\includegraphics{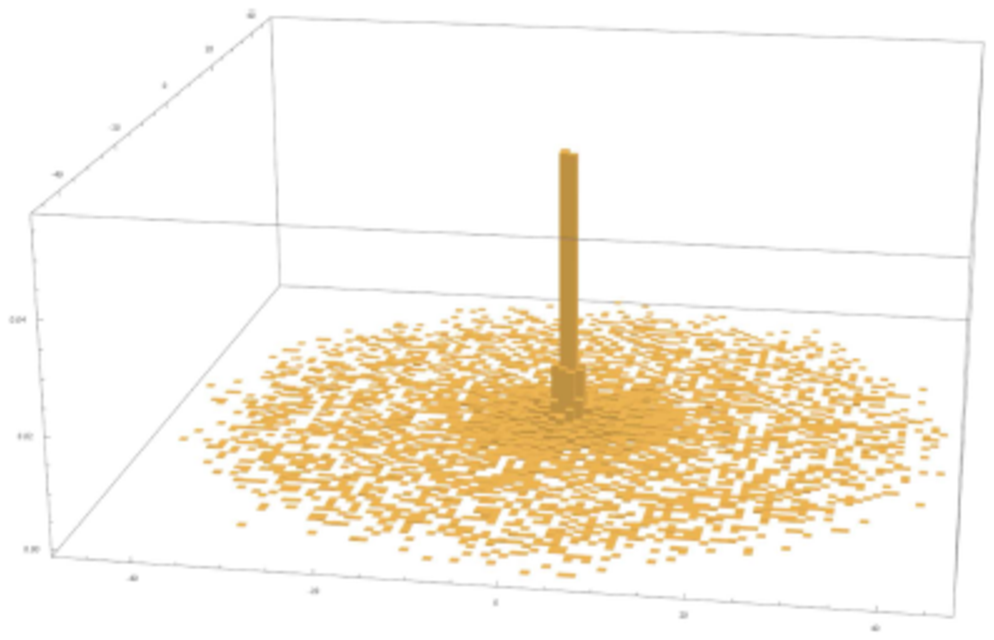}}
\caption{Three different variance values for the random entries. There appear to be three bulks ``stacked'' on top of each other.}
    \end{minipage}
\end{figure}

%\pagebreak

\appendix

%%%%%%%%%%%%%%%%%%%%%%%%%%%%%%%%%%%%%%%%%%%%%%%%%%%%%%%%%%%%%%%%%%%%%%%%%%%%%%%%%%%%%%%%%%%%%%%%%%%%%%%%%%%%%%%%
%%%%%%%%%%%%%%%%%%%%%%%%%%%%%%%%%%%%%%%%%%%%%%%%%%%%%%%%%%%%%%%%%%%%%%%%%%%%%%%%%%%%%%%%%%%%%%%%%%%%%%%%%%%%%%%%
%%%%%%%%%%%%%%%%%%%%%%%%%%%%%%%%%%%%%%%%%%%%%%%%%%%%%%%%%%%%%%%%%%%%%%%%%%%%%%%%%%%%%%%%%%%%%%%%%%%%%%%%%%%%%%%%
%%%%%%%%%%%%%%%%%%%%%%%%%%%%%%%%%%%%%%%%%%%%%%%%%%%%%%%%%%%%%%%%%%%%%%%%%%%%%%%%%%%%%%%%%%%%%%%%%%%%%%%%%%%%%%%%

\section{}
\subsection{Notation and terminology for cyclic products}
\label{appendix:terminology}
Throughout we have used some convenient terminology borrowed from \cite{RMT2016} to analyze the cyclic products. The definitions are copied and adapted below.

Recall that
\begin{equation}
\E\left[\mathrm{Tr}M^n\right] \ =\ \sum_{1 \leq i_1, \ldots, i_n \leq N} \E[m_{i_1 i_2}m_{i_2 i_3} \cdots m_{i_n i_1}].
\end{equation}
%It will become apparent that phrasing the arguments to come in terms of the following definitions drastically simplifies the discussions.  As such, we introduce the following vocabulary.
We refer to terms $\E[m_{i_1 i_2}m_{i_2 i_3} \cdots m_{i_n i_1}]$ as \textbf{cyclic products} and $m$'s as entries of cyclic products. Occasionally, some of our cyclic products appear in altered form, with certain terms $m_{i_j i_{j+1}}$ replaced instead with $m_{i_{j+1} i_{j}}$ or perhaps with complex conjugates $\overline{m_{i_j i_{j+1}}}$ or $\overline{m_{i_{j+1} i_{j}}}$, but we extend this terminology to those scenarios as well. In many of our moment arguments, we are interested in computing these cyclic products, which reduces to a combinatorics problem of understanding the contributions of different cyclic products. We develop the following vocabulary to classify types of cyclic products according to the aspects of their structure that determine overall contributions. % to break this combinatorics problem into several isolated, tractable steps. %The following definitions clarify later combinatorics.

\begin{definition}
A \textbf{term} refers to a single component $m_{i_j i_{j+1}}$ of the cyclic product.
\end{definition}

\begin{definition}\label{def of a Block}
A \textbf{block} is a set of adjacent $a$'s surrounded by $w$'s in a cyclic product, where the last entry of a cyclic product is considered to be adjacent to the first. We refer to a block of length $\ell$ as an $\ell$-block or sometimes a block of size $\ell$.
\end{definition}

\begin{definition}\label{def of a Configuration}
A \textbf{configuration} is the set of all cyclic products for which it is specified (a) how many blocks there are, and of what lengths, and (b) in what order these blocks appear. However, it is not specified how many $w$'s there are between each block.%at each position, it is specified in advance whether the entry at that position is an $a$ or a $w$, but the indices of each term in the product are not yet specified.
\end{definition}

\begin{example}
The set of all cyclic products of the form $w\cdots waw \cdots waaw \cdots waw \cdots w$, where each $\cdots$ represents a string of $w$'s and the indices are not yet specified, is a configuration.%$a_{ij}a_{jk}w_{ki}$, $a_{ij}w_{jk}a_{ki}$, and $w_{ij}a_{jk}a_{ki}$, where the values of $i$, $j$, and $k$ are not yet specified, is a $3$-configuration.
\end{example}

\begin{definition}\label{def of a Class}
Let $S$ be a multiset of natural numbers. An $S$\textbf{-class}, or class when $S$ is clear from context, is the set of all configurations for which there exists a unique $s$-block for every $s \in S$ counting multiplicity.  In other words, two configurations in the same class must have the same blocks but they may be ordered differently and have different numbers of $w$'s between them.
\end{definition}

%When we speak of the \emph{contribution} of a configuration or class to $\etr A^{2n+i}$, we assume that the length of the cyclic product is fixed at $2n+i$. The reason that the length of the cyclic product is suppressed in our notation is because $n(N)$ varies with $N$ and we wish to consider the contribution of a configuration or class as $N \rightarrow \infty$. %\fix{this is an important explanation, would welcome comments on moving or improving it.-R}

\begin{definition}\label{def of a Matching}
Given a configuration, a \textbf{matching} is an equivalence relation $\sim$ on the $a$'s in the cyclic product which constrains the ways of indexing (see Definition \ref{def of an Indexing}) the $a$'s as follows: an indexing of $a$'s conforms to a matching $\sim$ if, for any two $a$'s $a_{i_{\ell},i_{\ell+1}}$ and $a_{i_{t},i_{t+1}}$, we have $\{i_\ell,i_{\ell+1}\}=\{i_t,i_{t+1}\}$ if and only if $a_{i_{\ell}i_{\ell+1}} \sim a_{i_{t},i_{t+1}}$. We further constrain that each $a$ is matched with at least one other by any matching $\sim$.
\end{definition}
%$a_{ij} \sim a_{i'j'}$ if and only if for any indexing specified by the matching, $i = i'$ and $j = j'$.
%, formally, a map from the set of positions occupied by an $a$ to a finite set $Z$ such that the preimage of any element of $Z$ has size at least $2$.

\begin{remark}
Noting that the $a_{ij}$ are drawn from a mean-$0$ distribution, any matching with an unmatched $a$ would not contribute in expectation, hence it suffices to only consider those with the $a$'s matched at least in pairs.
\end{remark}

\begin{example}
Given a configuration $a_{i_1i_2}w_{i_2i_3}a_{i_3i_4}w_{i_4i_5}a_{i_5i_6}w_{i_6i_7}a_{i_7i_8}w_{i_8i_1}$ (the indices are not yet specified because this is a configuration), if $a_{i_1i_2} \sim a_{i_5i_6}$ we must have either $i_1=i_5$ and $i_2=i_6$ or $i_1=i_6$ and $i_2=i_5$.
\end{example}

\begin{definition}\label{def of an Indexing}
Given a configuration, matching, and length of the cyclic product, then an \textbf{indexing} is a choice of
\begin{enumerate}
\item the (positive) number of $w$'s between each pair of adjacent blocks (in the cyclic sense), and
\item the integer indices of each $a$ and $w$ in the cyclic product.
\end{enumerate}
\end{definition}

\begin{lemma}
\label{lemma:combinatorial}
(Lemma 3.16 from \cite{RMT2016}.) For any $0\leq p<m$
\begin{align}
\sum_{j=0}^m(-1)^j\binom{m}{j}j^p\ =  \ 0.
\end{align}
Furthermore
\begin{align}
\sum_{j=0}^m(-1)^{m-j}\binom{m}{j}j^m \  = \ m!.
\end{align}
\end{lemma}

\begin{comment}
We allow indexings to require fixing the length $\eta$ of the cyclic product because in the following counting scheme we split up our sums over cyclic products $ \E[a_{i_1 i_2}a_{i_2 i_3} \cdots a_{i_\eta i_1}]$ by classes, configurations, and matchings, then sum over indexings to contribute a factor of the form $K(k) \pfrac{N}{k}^{\eta - l}$ for $C$ and $l$ independent of $\eta$. Hence it is easy to consider the contribution of summing over indexings as $N \rightarrow \infty$.
\end{comment}

\subsection{Joint density for singular values of complex symmetric Gaussian ensemble}
\label{section:jointDensity}

We give a proof of the joint density of singular values for complex symmetric matrices found in \cite{AZ,Fo}.

\begin{theorem}[Joint Density of Singular Values for Complex Symmetric Matrices]
\label{thm:joint density}
Suppose $M$ is a random complex symmetric $N\times N$ matrix (not necessarily Hermitian), with entries in the upper triangle half and the diagonal iidrv mean $0$ variance $1$ complex Gaussian random variables. The joint density of the singular values of $M$ is given by
\begin{align}
\rho_N(x_1,\ldots,x_N)\ = \ c_N\left|\Delta(x_1^2,\ldots,x_N^2)\right|\prod_{j=1}^N |x_j| \prod_{j=1}^Ne^{-|x_j|^2/2}.
\end{align}
\end{theorem}

We adapt a proof of Ginibre's formula for the eigenvalue joint density of complex asymmetric matrices as presented by Stephen Ge \cite{Ge}, who cites Mehta \cite{Me,MeDy}.

\begin{proof}[Proof of Theorem \ref{thm:joint density}]
Let $|M|^2\ = \ {\rm Tr}(M^*M)$ denote the Frobenius (Hilbert-Schmidt) norm. Then
\begin{align}
dP\ := \ C_N e^{-|M|^2/2}~dM\ = \ C_N\prod_{i,j}e^{-|x_{ij}|^2/2}~dM \label{dP}
\end{align}
gives $M$'s density on the space of all $n\times n$ complex symmetric matrices, where $dM$ is Lebesgue measure on that space and $C_n$ is some constant.

We derive the desired formula by computing $\P(|M-D'|\leq\epsilon)$ in two ways, for $\varepsilon >0$ arbitrarily small and $D'$ a fixed diagonal matrix with non-negative real entries. The above density gives
\begin{align}
\P(|M-D'|\leq\epsilon)\ &= \ \int_{|M-D'|\leq\epsilon} ~dP.
%
%
%(C+o(1))e^{-|D'|^2}\epsilon^{n^2+n}.
\end{align} %
We can treat this integral as the volume of a thin rectangle centered at D', since $\epsilon$ is very small. Since we take $M$ from a distribution of complex symmetric matrices, and such matrices have $N^2 + N$ degrees of freedom, this volume is bounded above by $C_n e^{-|D'|^2/2} \epsilon^{N^2+N}$. Thus we have
\begin{align}
\P(|M-D'|\leq\epsilon)\ = \ (C+o(1))e^{-|D'|^2/2}\epsilon^{N^2+N}. \label{firstway}
\end{align}
%EDIT-- EXPLANATION ADDED

We now compute this probability a second way by using the Takagi factorization of a complex symmetric matrix $M$:
\begin{align}
M\ = \ UDU^\intercal, \label{takagi}
\end{align}
where $U$ is unitary and $D$ is a diagonal matrix with nonnegative real entries. Since unitary matrices $U$ can be written as $U=\exp(S)$ for some $S$ skew Hermitian, $U$ has $N^2$ degrees of freedom, while $D$ has $N$. Thus, the left and the right-hand sides of \eqref{takagi} have the same number of degrees of freedom.

%EDIT-- added d.o.f explanation
%The above density on $n\times n$ symmetric matrices induces

Define a density $\psi(D)dD$ on the space of diagonal matrices with nonnegative real entries, so that when $U$ is taken from the unitary group uniformly and $D$ with density $\psi(D)$, $M =  UDU^\intercal$ is a Gaussian random matrix. We eventually use both probability expressions to compute $\psi(D)$, which will in turn allow us to determine the joint density formula.

Now, let $M$ be such that $|M-D'|\ \leq \ \epsilon$. Then, following the lead of Tao \cite{Tao}, we write $U = I+O(\epsilon)$ and $D = D'+\varepsilon E$, where $E$ is real diagonal.
%added E description
After counting degrees of freedom, $S$ has density $C'(1+o(1))\epsilon^{N^2}dS$, where $dS$ is the Lebesgue measure on the space of skew-Hermitian matrices. Similarly, $E$ has density $C''(1+o(1))\epsilon^n\psi(D')dE$. We have
\begin{align}
M\ &= \ UDU^\intercal \nonumber
\\
&= \ \mathrm{exp}(\epsilon S)(D' +\epsilon E)\mathrm{exp}(\epsilon S)^\intercal \nonumber
\\
&= \ \mathrm{exp}(\epsilon S)(D' +\epsilon E)\mathrm{exp}(\epsilon S^\intercal) \nonumber
\\
&= \ \mathrm{exp}(\epsilon S)(D' +\epsilon E)\mathrm{exp}(-\epsilon \overline{S}).
\end{align}
Thus we can write
\begin{align}
\P(|M-D'|\ \leq \ \epsilon)\ &= \ C'''\int\int_{|\mathrm{exp}(\epsilon S)(D'+\epsilon E)\mathrm{exp}(-\epsilon \overline{S})-D'|\leq\epsilon} (1+o(1))\epsilon^{N^2}dS \epsilon^{N}\psi(D')~dE \nonumber
\\
&= \ C'''\epsilon^{N^2+N}(1+o(1))\psi(D')\int\int_{|\mathrm{exp}(\epsilon S)(D'+\epsilon E)\mathrm{exp}(-\epsilon \overline{S})-D'|\leq\epsilon} ~dSdE. \label{secondway}
\end{align}
Taylor-expanding the exponential to first order gives
\begin{eqnarray}
|\mathrm{exp}(\epsilon S)(D'+\epsilon E)\mathrm{exp}(-\epsilon \overline{S})-D'| & \ = \ & |(D'+\epsilon SD'+\epsilon E-\epsilon D'\overline{S}+O(\epsilon^2))-D'| \nonumber\\ &=& |\epsilon(E + SD' - D'\overline{S} + O(\epsilon))|.
\end{eqnarray}
which from above we want to be at most $\epsilon$. As $|B|^2 = {\rm Tr}(B^\ast B)$, this implies that
\begin{align}
|E+SD'-D'\overline{S}|\ \leq \  1+O(\epsilon).
\end{align}
Consider the change of variables
\begin{align}
A\ = \ E+SD'-D'\overline{S}.
\end{align}
Entry-by-entry, we have
\begin{align}
A_{jk}\ = \ E_{jk}+(D'_{kk}S_{jk}-D'_{jj}\overline{S_{jk}}).
\end{align}
Here we have used the fact that $D'$ is diagonal. Also, the skew-symmetry of $S$ and the diagonality of $E$ and $D'$ imply $A_{jk}$ is symmetric.

Next, we write the real and imaginary parts of this change of coordinates separately. Let $\sigma_{A_{jk}}$ denote the real part of $A_{jk}$, and $\tau_{A_{jk}}$ denote the imaginary part of $A_{jk}$. Recalling that $D$ is real diagonal, we have
\begin{align}
\sigma_{A_{jk}}\ &= \ \sigma_{E_{jk}}+\sigma_{S_{jk}}(\sigma_{D'_{kk}}-\sigma_{D'_{jj}})
\nonumber\\
\tau_{A_{jk}}\ &= \ \tau_{S_{jk}}(\sigma_{D'_{kk}}+\sigma_{D'_{jj}}).
\end{align}

We can interpret this matrix change of variables $A\rightarrow (S,E)$ as a transformation from $\R^{n^2+n}$ to $\R^{n^2+n}$, i.e., from (for $j\leq k$) the $\sigma_{A_{jk}}$ and $\tau_{A_{jk}}$ to $n^2+n$-tuples with entries $\sigma_{S_{jk}}$ (for $j<k$), $\tau_{S_{jk}}$ (for $j\leq k$), and $\sigma_{E_{jj}}$. This transformation is a direct sum of three diagonal transformations:
\begin{align}
\sigma_{A_{jk}}\ &= \ \sigma_{S_{jk}}(\sigma_{D'_{kk}}-\sigma_{D'_{jj}}) \text{\hspace{0.5 in} for $j\neq k$}
\nonumber\\
\sigma_{A_{jj}}\ &= \ \sigma_{E_{jj}}
\nonumber\\
\tau_{A_{jk}}\ &= \ \tau_{S_{jk}}(\sigma_{D'_{kk}}+\sigma'_{D'_{jj}}).
\end{align}
This is a \textit{diagonal} transformation, i.e., its Jacobian is the product of each of the above scaling factors. The Jacobian for the change of coordinates $(S,E)\rightarrow A$ is
\begin{align}
\prod_{1\leq j<k\leq n} |D'^2_{kk}-D'^2_{jj}|\prod_{j=1}^n |2D'_{jj}|\ = \ |\Delta(D'^2_{11},\ldots,D'^2_{nn})|^{-1}\prod_{j=1}^n |2D'_{jj}|^{-1}.
\end{align}
Returning to \eqref{secondway}, we have
\begin{align}
\P(|M-D'|\leq\epsilon)\ &= \ C'''\epsilon^{N^2+N}(1+o(1))\psi(D')\int\int_{|\mathrm{exp}(\epsilon S)(D'+\epsilon E)\mathrm{exp}(-\epsilon \overline{S})-D'|\leq\epsilon} ~dSdE \nonumber
\\
&= \ C'''\epsilon^{N^2+N}(1+o(1))\psi(D')\int_{|A|\leq1+O(\epsilon)}~dA.
\end{align}

Note that as $\epsilon\rightarrow0$, the integral in $A$ goes to a constant. Absorbing this into the constant $C'''$ and comparing the above expression for $\P(|M-D'|\ \leq \ \epsilon)$ with \eqref{firstway} gives
\begin{align}
\psi(D')\ &= \ C''''|\Delta(D_{11}^2,\ldots,D_{nn}^2)|\left(\prod_{j=1}^N |2D'_{jj}|\right) e^{-|D'|^2/2} \nonumber
\\
&= \ C''''|\Delta(D_{11}^2,\ldots,D_{nn}^2)|\prod_{j=1}^N |2D'_{jj}|\prod_{j=1}^N e^{-D'^2_{jj}/2}.
\end{align}
As the diagonal entries $D_{jj}$ are precisely the singular values of $M$, we find
\begin{align}
\rho_n(x_1,\ldots,x_N)\ = \ C'''' \left |\Delta (x_1^2,\ldots,x_N^2) \right | \prod_{j=1}^N \left |2x_j \right | \prod_{j=1}^N e^{-|x_j|^2/2},
\end{align}
which proves the theorem after absorbing $2^N$ into the constant.
\end{proof}

Theorem \ref{thm:joint density} allows us to compute the distribution of the least singular value for the Gaussian complex symmetric ensemble.
As before, we list the available distribution for the least singular values of the complex asymmetric Gaussian ensemble as computed by Edelman \cite{Edelman1988}:
\begin{align}
&{\rm Complex\ asymmetric\ Gaussian:} \ p(\sigma_N)\ =\ N\sigma_N e^{-N\sigma_N^2/2}
\end{align}

\begin{corollary}
\label{cor:least singVal}
The probability density function of the least singular value $\sigma_N$ follows the Rayleigh distribution
\begin{align}
p(\sigma_N)\ =\ N\sigma_N e^{-N\sigma_N^2/2}.
\end{align}
\end{corollary}
Note that, while the joint densities are distinct, the least singular value of the complex asymmetric and complex symmetric Gaussian matrices share the same distribution.

\begin{proof}[Proof of Corollary \ref{cor:least singVal}]
Order the singular values $0\leq\sigma_N\leq\cdots\leq\sigma_1$. Using Theorem \ref{thm:joint density}, we can integrate out the other parameters $\sigma_1,\ldots,\sigma_{N-1}$ from the joint density to obtain the density function of the least singular value:
\begin{align}
p(\sigma_N) &\ =\ \int_{\sigma_N\leq\sigma_{N-1}\leq\cdots\leq\sigma_1} \rho_N(\sigma_1,\sigma_2,\ldots,\sigma_N) ~d\sigma_1\cdots d\sigma_{N-1} \nonumber
\\
&\ =\ C_N\int_{\sigma_N\leq\sigma_{N-1}\leq\cdots\leq\sigma_1}  \prod_{1\leq k<j \leq N} \left(\sigma_k^2-\sigma_j^2\right)\prod_{j=1}^N \sigma_j \prod_{j=1}^N e^{-\sigma_j^2/2} ~d\sigma_1\cdots d\sigma_{N-1} \nonumber
\\
&\ =\ C_N\sigma_N e^{-N \sigma_N^2/2}
\int_{\sigma_N\leq\sigma_{N-1}\leq\cdots\leq\sigma_1} \prod_{1\leq k<j \leq N-1} \left(\sigma_k^2-\sigma_j^2\right) \prod_{1\leq k\leq N-1} \left(\sigma_k^2-\sigma_N^2\right) \nonumber\\ & \ \ \ \ \ \ \ \ \ \ \ \cdot \ \prod_{j=2}^n \sigma_j \prod_{j=2}^n e^{-(\sigma_j^2-\sigma_N^2)/2} ~d\sigma_1\cdots d\sigma_{N-1}. \label{lambda_fact}
\end{align}

As is done in \cite{Edelman1988}, make the change of variables $x_j=\sigma_j^2-\sigma_N^2$. Thus, with $dx_j=2d\sigma_j$, we find
\begin{align}
p(\sigma_N) &=2^{1-N}C_N\sigma_N e^{-N\sigma_N^2/2} \int_{0\leq x_{N-1}\leq\cdots\leq x_1} \prod_{1\leq k<j \leq N-1} (x_k-x_j) \prod_{j=1}^{N-1} x_j \prod_{j=1}^{N-1} e^{-x_j/2} ~dx_1 \cdots dx_{N-1}.
\end{align}
Since the integral is independent of $\sigma_N$, it is constant, which we denote by $C_N'$:
\begin{align}
p(\sigma_N)\ =\ C_N'\sigma_N e^{-N\sigma_N^2/2}.
\end{align}
As $p$ is a probability distribution, $\int_{\R_{\geq 0}} p(\sigma_N)~d\sigma_N =1$ and thus $C_N'=N$, completing the proof.
\end{proof}

%--------------------------------%
%--------------------------------%
%%%%%%%%%%%%%%%%%%%%%%%%%%%%%%%%%%%%%%%%%%%%%%%%%%%%%%%%%%%%%%%%%%%%%%%%%%%%%%%%%%%%%%%%%%%%%%%%%%%%%%%%%%%%%%%%
%%%%%%%%%%%%%%%%%%%%%%%%%%%%%%%%%%%%%%%%%%%%%%%%%%%%%%%%%%%%%%%%%%%%%%%%%%%%%%%%%%%%%%%%%%%%%%%%%%%%%%%%%%%%%%%%
%%%%%%%%%%%%%%%%%%%%%%%%%%%%%%%%%%%%%%%%%%%%%%%%%%%%%%%%%%%%%%%%%%%%%%%%%%%%%%%%%%%%%%%%%%%%%%%%%%%%%%%%%%%%%%%%
%%%%%%%%%%%%%%%%%%%%%%%%%%%%%%%%%%%%%%%%%%%%%%%%%%%%%%%%%%%%%%%%%%%%%%%%%%%%%%%%%%%%%%%%%%%%%%%%%%%%%%%%%%%%%%%%

%\clearpage
%\bibliographystyle{amsplain}
%\bibliography{RMT}

\providecommand{\noopsort}[1]{}
\providecommand{\bysame}{\leavevmode\hbox to3em{\hrulefill}\thinspace}
\providecommand{\MR}{\relax\ifhmode\unskip\space\fi MR }
\providecommand{\MRhref}[2]{%
  \href{http://www.ams.org/mathscinet-getitem?mr=#1}{#2}
}

%%%%%%%%%%%%%%%%%%%%%%%%%%%%%%%%%%%%%5%%%%%%%%%%%%%%%%%%%%%%%%%%%%%%%%%%%%%5%%%%%%%%%%%%%%%%%%%%%%%%%%%%%%%%%%%%%5
%%%%%%%%%%%%%%%%%%%%%%%%%%%%%%%%%%%%%5%%%%%%%%%%%%%%%%%%%%%%%%%%%%%%%%%%%%%5%%%%%%%%%%%%%%%%%%%%%%%%%%%%%%%%%%%%%5
%%%%%%%%%%%%%%%%%%%%%%%%%%%%%%%%%%%%%5%%%%%%%%%%%%%%%%%%%%%%%%%%%%%%%%%%%%%5%%%%%%%%%%%%%%%%%%%%%%%%%%%%%%%%%%%%%5
%%%%%%%%%%%%%%%%%%%%%%%%%%%%%%%%%%%%%5%%%%%%%%%%%%%%%%%%%%%%%%%%%%%%%%%%%%%5%%%%%%%%%%%%%%%%%%%%%%%%%%%%%%%%%%%%%5
\providecommand{\href}[2]{#2}

\ \\

%---------------------------------%

\end{document}